\documentclass{lmcs}
\pdfoutput=1

\usepackage{lastpage}
\lmcsdoi{16}{4}{13}
\lmcsheading{}{\pageref{LastPage}}{}{}%
{Dec.~12,~2019}{Dec.~14,~2020}{}

\keywords{Markov decision processes, window mean-payoff, window parity}

\usepackage[utf8]{inputenc}

\usepackage{microtype}
\usepackage[english]{babel}
\usepackage{graphicx}
\usepackage{tikz}
\usetikzlibrary{decorations,arrows,shapes,automata,calc}
\usepackage{algorithm}
\usepackage[noend]{algorithmic}
\usepackage{wrapfig}
\usepackage{hhline}
\usepackage{multirow}
\newcommand{\mdp}{\mathcal{M}}
\newcommand{\paritySub}{\mathsf{par}}
\newcommand{\mpSub}{\mathsf{mp}}
\newcommand{\GW}{\mathsf{GW}}
\newcommand{\DFW}{\mathsf{DFW}}
\newcommand{\DBW}{\mathsf{DBW}}
\newcommand{\FW}{\mathsf{FW}}
\newcommand{\BW}{\mathsf{BW}}

\newcommand{\sureStrat}[3]{\mathsf{S}_{#1}^{#2}[#3]}
\newcommand{\almostStrat}[3]{\mathsf{AS}_{#1}^{#2}[#3]}

\newcommand{\cyl}{\mathsf{Cyl}}

\newcommand{\set}[1]{\{#1\}}

\newcommand{\calC}{\mathcal{C}}

\newcommand{\lmdp}{\mdp_{\lambda}}
\newcommand{\tS}{\tilde{S}}
\newcommand{\tdelta}{\tilde{\delta}}
\newcommand{\reach}{\mathsf{Reach}}
\newcommand{\safety}{\mathsf{Safety}}
\newcommand{\cobuchi}{\mathsf{coBuchi}}
\newcommand{\buchi}{\mathsf{Buchi}}

\newcommand{\states}{\ensuremath{S} }
\newcommand{\state}{\ensuremath{s} }
\newcommand{\supp}{\ensuremath{{\sf Supp}} }

\newcommand{\dists}{\ensuremath{\mathcal{D}} }

\newcommand*{\pr}{\mathbb{P}}

\newcommand{\Inf}{\mathsf{inf}}
\newcommand{\InfAct}{\mathsf{infAct}}
\newcommand{\InfRun}{\mathsf{limitSet}}

\newcommand{\strats}{\ensuremath{\Sigma} }

\newcommand\restr[2]{\ensuremath{\left.#1\right|_{#2}}}
\newcommand\mecs{\ensuremath{\mathsf{MEC}}}
\newcommand\ecs{\ensuremath{\mathsf{EC}}}
\newcommand{\Runs}[1]{\ensuremath{\mathsf{Runs}(#1)}}
\newcommand{\Hists}[1]{\ensuremath{\mathsf{Hists}(#1)}}
\newcommand\NPTIME{\mathsf{NP}}

\newcommand\UPTIME{\mathsf{UP}}
\newcommand\coUPTIME{\mathsf{coUP}}
\newcommand\PTIME{\mathsf{P}}
\newcommand\LOGSPACE{\mathsf{LOGSPACE}}

\begin{document}

\title[Life is Random, Time is Not: MDPs with Window Objectives]{Life is Random, Time is Not:\texorpdfstring{\\}{}Markov Decision Processes with Window Objectives\texorpdfstring{\rsuper*}{}}
\titlecomment{{\lsuper*}Research supported by F.R.S.-FNRS under Grant n$^\circ$ F.4520.18 (ManySynth), and F.R.S.-FNRS mobility funding for scientific missions (Y.~Oualhadj in UMONS, 2018). Florent Delgrange is now affiliated with Vrije Universiteit Brussel, Belgium. Mickael Randour is an F.R.S.-FNRS Research Associate.}

\author[T.~Brihaye]{Thomas Brihaye\rsuper{a}}

\author[F.~Delgrange]{Florent Delgrange\rsuper{{b,a}}}

\author[M.~Randour]{\texorpdfstring{\\}{}Mickael Randour\rsuper{{c,a}}}

\author[Y.~Oualhadj]{Youssouf Oualhadj\rsuper{d}\texorpdfstring{\vspace{-5mm}}{}}

\address{\lsuper{a}UMONS --- Université de Mons, Belgium}
\address{\lsuper{b}RWTH Aachen, Germany}
\address{\lsuper{c}F.R.S.-FNRS, Belgium}
\email{mickael.randour@umons.ac.be}
\address{\lsuper{d}LACL --- UPEC, France}

\begin{abstract}
The window mechanism was introduced by Chatterjee et al.~to strengthen classical game objectives with time bounds. It permits to synthesize system controllers that exhibit acceptable behaviors within a configurable time frame, all along their infinite execution, in contrast to the traditional objectives that only require correctness of behaviors in the limit. The window concept has proved its interest in a variety of two-player zero-sum games because it enables reasoning about such time bounds in system specifications, but also thanks to the increased tractability that it usually yields.

In this work, we extend the window framework to stochastic environments by considering Markov decision processes. A fundamental problem in this context is the \textit{threshold probability problem}: given an objective it aims to synthesize strategies that guarantee satisfying runs with a given probability.  We solve it for the usual variants of window objectives, where either the time frame is set as a parameter, or we ask if such a time frame exists. We develop a generic approach for window-based objectives and instantiate it for the classical mean-payoff and parity objectives, already considered in games. Our work paves the way to a wide use of the window mechanism in stochastic models.
\end{abstract}

\maketitle

\section{Introduction}%
\label{sec:intro}

\paragraph{Game-based models for controller synthesis}  \textit{Two-player zero-sum games}~\cite{DBLP:conf/dagstuhl/2001automata,rECCS} and \textit{Markov decision processes} (MDPs)~\cite{filar1997,baier2008principles,DBLP:conf/vmcai/RandourRS15} are two popular frameworks to model decision making in adversarial and uncertain environments respectively. In the former, a system controller and its environment compete antagonistically, and synthesis aims at building strategies for the controller that ensure a specified behavior \textit{against all possible strategies of the environment}. In the latter, the system is faced with a given stochastic model of its environment, and the focus is on satisfying a given level of expected performance, or a \textit{specified behavior with a sufficient probability}. Classical objectives studied in both settings notably include \textit{parity}, a canonical way of encoding $\omega$-regular specifications, and \textit{mean-payoff}, which evaluates the average payoff per transition in the limit of an infinite run in a weighted graph.

\paragraph{Window objectives in games} The traditional parity and mean-payoff objectives share two shortcomings. First, they both reason about infinite runs \textit{in their limit}. While this elegant abstraction yields interesting theoretical properties and makes for robust interpretation, it is often beneficial in practical applications to be able to specify a \textit{parameterized time frame} in which an acceptable behavior should be witnessed. Second, both parity and mean-payoff games belong to $\UPTIME \cap \coUPTIME$~\cite{DBLP:journals/ipl/Jurdzinski98,DBLP:conf/cav/GawlitzaS09}, but despite recent breakthroughs~\cite{DBLP:conf/stoc/CaludeJKL017,DBLP:conf/lics/DaviaudJL18}, they are still not known to be in~$\PTIME$. Furthermore, the latest results~\cite{DBLP:journals/corr/abs-1807-10546,mpuniv} indicate that all existing algorithmic approaches share inherent limitations that prevent inclusion in $\PTIME$.

Window objectives address the time frame issue as follows. In their \textit{fixed} variant, they consider a window of size bounded by $\lambda \in \mathbb{N}_0$ (given as a parameter) sliding over an infinite run and declare this run to be winning if, in all positions, the window is such that the (mean-payoff or parity) objective is locally satisfied. In their \textit{bounded} variant, the window size is not fixed a priori, but a run is winning if there exists a bound $\lambda$ for which the condition holds. Window objectives have been considered both in \textit{direct} versions, where the window property must hold from the start of the run, and \textit{prefix-independent} versions, where it must hold from some point on. Window games were initially studied for mean-payoff~\cite{DBLP:journals/iandc/Chatterjee0RR15} and parity~\cite{DBLP:journals/corr/BruyereHR16}. They have since seen diverse extensions and applications: e.g.,~\cite{DBLP:conf/csl/BaierKKW14,DBLP:conf/rp/Baier15,DBLP:conf/concur/BrazdilFKN16,DBLP:conf/concur/BruyereHR16,DBLP:journals/acta/HunterPR18,DBLP:conf/fsttcs/0001PR18}.

\paragraph{Window objectives in MDPs} Our goal is to lift the theory of window games to the stochastic context. With that in mind, we consider the canonical \textit{threshold probability problem}: given an MDP, a window objective defining a set of acceptable runs $E$, and a probability threshold $\alpha$, we want to decide if there exists a controller \textit{strategy} (also called \textit{policy}) to achieve $E$ with probability at least $\alpha$. It is well-known that many problems in MDPs can be reduced to threshold problems for appropriate objectives: e.g., maximizing the \textit{expectation} of a prefix-independent function boils down to maximizing the probability to reach the best end-components for that function (see examples in~\cite{baier2008principles,DBLP:journals/iandc/BruyereFRR17,DBLP:journals/fmsd/RandourRS17}).

\paragraph{Example}  Before going further, let us consider an example. Take the MDP depicted in Fig.~\ref{fig:intro_ex}: circles depict states and dots depict actions, labeled by letters. Each action yields a probability distribution over successor states: for example, action $b$ leads to $s_2$ with probability $0.5$ and $s_3$ with the same probability. This MDP is actually a \textit{Markov chain} (MC) as the controller has only one action available in each state: this process is purely stochastic.

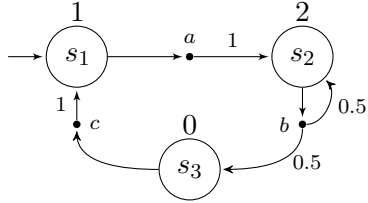
\begin{figure}[tbh]
  \begin{center}
  \begin{tikzpicture}[->,>=stealth',shorten >=1pt,auto,node
    distance=2.5cm,bend angle=45, scale=0.6, font=\normalsize]
    \tikzstyle{p1}=[draw,circle,text centered,minimum size=7mm,text width=4mm]
    \tikzstyle{p2}=[draw,rectangle,text centered,minimum size=7mm,text width=4mm]
    \tikzstyle{act}=[fill,circle,inner sep=1pt,minimum size=1.5pt, node distance=1cm]    \tikzstyle{empty}=[text centered, text width=15mm]
    \node[p1] (1) at (0,0) {$s_{1}$};
    \node[empty] at ($(1)+(0,1)$) {1};
    \node[p1] (2) at (5,0) {$s_{2}$};
    \node[empty] at ($(2)+(0,1)$) {2};
    \node[p1] (3) at (2.5,-2.5) {$s_{3}$};
    \node[empty] at ($(3)+(0,1)$) {0};
    \node[act] (1a) at (2.5,0) {};
    \node[empty] at ($(1a)+(0,0.4)$) {\scriptsize $a$};
    \node[act] (2a) at (5,-1.5) {};
    \node[empty] at ($(2a)-(0.4,0)$) {\scriptsize $b$};
    \node[act] (3a) at (0,-1.5) {};
    \node[empty] at ($(3a)+(0.4,0)$) {\scriptsize $c$};
    \coordinate[shift={(-5mm,0mm)}] (init) at (1.west);
    \path[-latex']
    (init) edge (1)
	(1) edge (1a)
	(2) edge (2a)
    (1a) edge node[above,xshift=0mm]{\scriptsize $1$} (2)
    (3a) edge node[left,xshift=0mm]{\scriptsize $1$}  (1)
    ;
	\draw [->] (2a) to[out=0,in=315] node[right,xshift=0mm]{\scriptsize $0.5$} (2);
	\draw [->] (2a) to[out=270,in=0] node[right,xshift=1mm]{\scriptsize $0.5$} (3);
	\draw [->] (3) to[out=180,in=270]  (3a);
      \end{tikzpicture}
    \caption{Simple Markov chain where parity is surely satisfied but all window parity objectives have probability zero.}%
    \label{fig:intro_ex}
  \end{center}
\end{figure}

We consider the \textit{parity} objective here: we associate a non-negative integer priority with each state, and a run is winning if the minimum one amongst those seen infinitely often is even.
Clearly, any run in this MC is winning: either it goes through $s_3$ infinitely often and the minimum priority is $0$, or it does not, and the minimum priority seen infinitely often is~$2$. Hence the controller not only wins \textit{almost-surely} (with probability one), but even \textit{surely} (on all runs).

Now, consider the \textit{window parity} objective that informally asks for the minimum priority inside a window of size bounded by $\lambda$ to be even, with this window sliding all along the infinite run. Fix any $\lambda \in \mathbb{N}_0$. It is clear that every time $s_1$ is visited, there will be a fixed strictly positive probability $\varepsilon > 0$ of not seeing $0$ before $\lambda$ steps: this probability is $1/2^{\lambda-1}$. Let us call this seeing a \textit{bad window}. Since we are in a \textit{bottom strongly connected component} of the MC, we will almost-surely visit $s_1$ infinitely often~\cite{baier2008principles}. Using classical probability arguments (Borel-Cantelli), one can easily be convinced that the probability to see bad windows infinitely often is one. Hence the probability to win the window parity objective is zero. This canonical example illustrates the difference between traditional parity and window parity: the latter is more restrictive as it asks for a strict bound on the time frame in which each odd priority should be answered by a smaller even priority.

Note that in practice, such a behavior is often wished for. For example, consider a computer server having to grant requests to clients. A classical parity objective can encode that requests should eventually be granted. However, it is clear that in a desired controller, requests should not be placed on hold for an arbitrarily long time. The existence of a finite bound on this holding time can be modeled with a \textit{bounded} window parity objective, while a specific bound can also be set as a parameter using a \textit{fixed} window parity objective.

\paragraph{Our contributions} We study the \textit{threshold probability problem} in MDPs for window objectives based on \textit{parity} and \textit{mean-payoff}, two prominent formalisms in qualitative and quantitative (resp.) analysis of systems. We consider the different variants of window objectives mentioned above: fixed vs.~bounded, direct vs.~prefix-independent. A nice feature of our approach is that we provide a \textit{unified view} of parity and mean-payoff window objectives: \textit{our algorithm can actually be adapted for any window-based objective} if an appropriate black-box is provided for a restricted sub-problem. This has two advantages: (i) conceptually, our approach permits a \textit{deeper understanding of the essence of the window mechanism}, not biased by technicalities of the specific underlying objective; (ii) our framework can \textit{easily} be \textit{extended} to other objectives for which a window version could be defined. This point is of great practical interest too, as it opens the perspective of a modular, generic software tool suite for window objectives.

\renewcommand{\arraystretch}{1.3}
\begin{table}
\small
\begin{center}
\begin{tabular}{|c||c|c||c|c|}
\cline{2-5}
\multicolumn{1}{c|}{} & \multicolumn{2}{c||}{parity} & \multicolumn{2}{c|}{mean-payoff} \\
\cline{2-5}
\multicolumn{1}{c|}{} & \;complexity\; & memory & complexity & memory\\
\hhline{-|====|}
\;DFW\; & \multirow{3}{*}{$\mathsf{P}$-c.} & \multirow{2}{*}{\;polynomial\;}  &  \;$\mathsf{EXPTIME}$/$\mathsf{PSPACE}$-h.\; & \;pseudo-polynomial\; \\
\cline{1-1}\cline{4-5}
\;FW\; & & & $\mathsf{P}$-c. & polynomial \\
\cline{1-1} \cline{3-5}
\;BW\; & & \;memoryless\; & $\mathsf{NP} \cap \mathsf{coNP}$ & memoryless \\
\hline
\end{tabular}
\end{center}
\caption{Complexity of the threshold probability problem for window objectives in MDPs and memory requirements. Acronyms DFW, FW and BW respectively stand for direct fixed window, fixed window and bounded window objectives. All memory bounds are tight and \textit{pure} strategies always suffice. For mean-payoff, the $\mathsf{PSPACE}$-hardness holds even for acyclic MDPs, and the bounded case is as hard as mean-payoff games. All results are new.}%
\label{table:overview}
\end{table}

We give an overview of our results in Table~\ref{table:overview}. For the sake of space, we use acronyms below: DFW for direct fixed window, FW for (prefix-independent) fixed window, DBW for direct bounded window, and BW for (prefix-independent) bounded window. Our main contributions are as follows.
\begin{enumerate}
\item We solve DFW MDPs through reductions to safety MDPs over \textit{well-chosen unfoldings}. This results in polynomial-time and pseudo-polynomial-time algorithms for the parity and mean-payoff variants respectively (Thm.~\ref{thm:direct}). We prove these complexities to be almost tight (Thm.~\ref{thm:directHardness}), the most interesting case being the $\mathsf{PSPACE}$-hardness of DFW mean-payoff objectives, even in the case of \textit{acyclic} MDPs.

We also show that no upper bound can be established on the window size needed to win in general (Sect.~\ref{subsec:illustration}), \textit{in stark contrast to the two-player games situation} (Sect.~\ref{subsec:games}).

\item We use similar reductions to prove that \textit{finite memory} suffices in the prefix-independent case (Thm.~\ref{finite-memory-cor}). Yet, in this case, we can do better than using the unfoldings to solve the problem. We start by studying \textit{end-components} (ECs), the crux for all prefix-independent objectives in MDPs: we show that ECs can be classified based on their \textit{two-player zero-sum game interpretation} (Sect.~\ref{sec:ec}). Using the result on finite memory, we prove that in ECs classified as \textit{good}, almost-sure satisfaction of window objectives can be ensured, whereas it is impossible to satisfy them with non-zero probability in the other ECs (Lem.~\ref{grolem}). We also establish tight complexity bounds for this classification problem (Thm.~\ref{thm:classification}). This \textit{EC classification} is --- both conceptually and complexity-wise --- the cornerstone to deal with general MDPs.

\item Our general algorithm is developed in Sect.~\ref{sec:general}: we prove $\mathsf{P}$-completeness for all prefix-independent variants but for the BW mean-payoff one (Thm.~\ref{thm:general} and Thm.~\ref{thm:generalHardness}), where we show that the problem is in $\mathsf{NP} \cap \mathsf{coNP}$ and as hard as mean-payoff games, a canonical ``hard'' problem for that complexity class.

\item For all variants, we prove \textit{tight memory bounds}: see Thm.~\ref{thm:direct} and Thm.~\ref{thm:directHardness} for the direct fixed variants, Thm.~\ref{thm:general} and Thm.~\ref{thm:generalHardness} for the prefix-independent fixed and bounded ones. In all cases, \textit{pure} strategies (i.e., without randomness) suffice.

\item We leave out DBW objectives from our analysis as we show they are not well-behaved. We illustrate their behavior in Sect.~\ref{subsec:illustration} and discuss their pitfalls in Sect.~\ref{sec:concl}.
\end{enumerate}
Along the way, we develop several side results that help drawing a \textit{line between MDPs and games} w.r.t.~window objectives: e.g., whether or not a uniform bound exists in the bounded case (Sect.~\ref{subsec:illustration}). As stated above, our approach is also generic and may be easily extended to other window-based objectives.

\paragraph{Comments on our results} In the game setting, window objectives are all in polynomial time, except for the BW mean-payoff variant, in $\mathsf{NP} \cap \mathsf{coNP}$. Despite clear differences in behaviors, the situation is almost the same here. The only outlier case is the DFW mean-payoff one, whose complexity rises significantly. As we show in Thm.~\ref{thm:directHardness}, the loss of prefix-independence permits to emulate shortest path problems on MDPs that are famously hard to solve efficiently (e.g.,~\cite{DBLP:conf/icalp/HaaseK15,DBLP:journals/fmsd/RandourRS17,DBLP:journals/corr/abs-1809-03107,DBLP:conf/tacas/HartmannsJKQ18}). In the almost-sure case, however, DFW mean-payoff MDPs collapse to $\mathsf{P}$ (Rmk.~\ref{rmk:DFWMP_AS_P}).

In games, window objectives permit to avoid long-standing $\mathsf{NP} \cap \mathsf{coNP}$ complexity barriers for parity~\cite{DBLP:journals/corr/BruyereHR16} and mean-payoff~\cite{DBLP:journals/iandc/Chatterjee0RR15}. Since both are known to be in $\mathsf{P}$ for the threshold probability problem in MDPs~\cite{DBLP:conf/soda/ChatterjeeJH04,DBLP:journals/fmsd/RandourRS17}, the main interest of window objectives resides in their \textit{modeling power}. Still, they may turn out to be more efficient in practice too, as polynomial-time algorithms for parity and mean-payoff, based on linear programming, are often forsaken in favor of exponential-time value or strategy iteration ones (e.g.,~\cite{DBLP:conf/cav/AshokCDKM17}).

\paragraph{Related work} We already mentioned many related articles, hence we only briefly discuss some remaining models here. Window parity games are strongly linked to the concept of \textit{finitary $\omega$-regular games}: see, e.g.,~\cite{DBLP:journals/tocl/ChatterjeeHH09}, or~\cite{DBLP:journals/corr/BruyereHR16} for a complete list of references. The window mechanism can be used to ensure a certain form of (local) guarantee over runs: different techniques have been considered in MDPs, notably \textit{variance-based}~\cite{DBLP:journals/jcss/BrazdilCFK17} or \textit{worst-case-based}~\cite{DBLP:journals/iandc/BruyereFRR17,DBLP:conf/icalp/BerthonRR17} methods.

Finally, let us mention the very recent work of Bordais et al.~\cite{DBLP:journals/corr/abs-1812-09298}, which considers a \textit{seemingly} related question: the authors define a \textit{value function} based on the window \textit{mean-payoff} mechanism and consider maximizing its expected value (which is different from the expected window size we discuss in Sect.~\ref{sec:concl}). While there definitively are similarities in our works w.r.t.~technical tools, the two approaches are quite different and have their own strengths: we focus on deep understanding of the window mechanism through a \textit{generic approach} for the canonical threshold probability problem for all window-based objectives, here instantiated as mean-payoff \textit{and} parity; whereas Bordais et al.~focus on a particular \textit{optimization problem} for a function relying on this mechanism.

In addition to having \textit{different philosophies}, our divergent approaches yield \textit{interesting differences}. We mention three examples illustrating the conceptual gap. First, in~\cite{DBLP:journals/corr/abs-1812-09298}, the studied function takes the same value for direct and prefix-independent bounded window mean-payoff objectives, whereas we show in Sect.~\ref{subsec:illustration} that the classical definitions of window objectives induce a striking difference between both (in the MDP of Ex.~\ref{ex:uniform}, the prefix-independent version is satisfied for window size one, whereas no uniform bound on all runs can be defined for the direct case). Second, we are able to prove $\mathsf{PSPACE}$-hardness for the DFW mean-payoff case, whereas the best lower bound known for the related problem in~\cite{DBLP:journals/corr/abs-1812-09298} is $\mathsf{PP}$. Lastly, let us recall that our work also deals with window \textit{parity} objectives while the function of~\cite{DBLP:journals/corr/abs-1812-09298} is strictly built on mean-payoff.

Our paper presents in full details, with additional proofs and examples, the contributions published in a preceding conference version~\cite{concur_sw}.

\paragraph{Outline} Sect.~\ref{sec:prelim} defines the model and problem under study. In Sect.~\ref{sec:window}, we introduce window objectives, discuss their status in games, and illustrate their behavior in MDPs. Sect.~\ref{sec:fixed} is devoted to the fixed variants and the aforementioned reductions. In Sect.~\ref{sec:ec}, we analyze the case of ECs and develop the classification procedure. We build on it in Sect.~\ref{sec:general} to solve the general case. Finally, in Sect.~\ref{sec:concl}, we discuss the limitations of our work, as well as interesting extensions within arm's reach (e.g., multi-objective threshold problem, expected value problem).

\section{Preliminaries}%
\label{sec:prelim}
\paragraph{Probability distributions} Given a set $S$, let~$\dists(S)$
denote the set of rational probability distributions over~$S$. Given a distribution $\iota \in \dists(S)$, let   $\supp(\iota) = \{ s \in S \mid \iota(s) > 0\}$ denote its support.

\paragraph{Markov decision processes} A finite
\textit{Markov decision process} (MDP) is a tuple
$\mdp = (S,A,\delta)$ where $S$ is a finite set of \emph{states},
$A$~is a finite set of \emph{actions} and
$\delta\colon S\times A \rightarrow \dists(S)$ is a partial function
called the \emph{probabilistic transition function}.  The
set of actions that are available in a state $\state \in \states$ (i.e., for which $\delta(s, a)$ is defined) is
denoted by $A(s)$.  We use $\delta(s,a,s')$ as a shorthand for
$\delta(s,a)(s')$. For the sake of readability, we write $(s,a,s') \in \delta$ to indicate that $\delta$ is defined on $(s,a)$ and $s' \in \supp(\delta(s,a))$.
We assume w.l.o.g.~that MDPs are \textit{deadlock-free}: for all
$s \in S$, $A(s) \neq \emptyset$. An MDP where for all
$s \in \states$, $\vert A(s)\vert = 1$ is a fully-stochastic process
called a \textit{Markov chain} (MC).

  A \emph{run}
of~$\mdp$ is an infinite sequence
$\rho = s_0a_0 \ldots a_{n-1} s_n\ldots{}$ of states and actions such
that $\delta(s_i,a_i,s_{i+1})>0$ for all~$i\geq 0$.  The
\textit{prefix} up to the $n$-th state of $\rho$ is the finite
sequence $\rho[0, n] = s_0a_0 \ldots a_{n-1}s_n$.  The \textit{suffix}
of $\rho$ starting from the $n$-th state of $\rho$ is the run
$\rho[n, \infty] = s_{n}a_{n}s_{n+1}a_{n+1}\dots$.  Moreover, we denote by
$\rho[n]$ the $n$-th state $s_n$ of $\rho$.  Finite prefixes of runs
of the form $h = s_0a_0 \ldots a_{n-1} s_n$ are called
\emph{histories}. We sometimes denote the last state of history $h$ by $\mathsf{Last}(h)$. We resp.~denote the sets of runs and histories of an
MDP $\mdp$ by $\Runs{\mdp}$ and $\Hists{\mdp}$.

\paragraph{End-components} Fix an MDP $\mdp = (S,A,\delta)$.  A \textit{sub-MDP} of
$\mdp$ is an MDP $\mdp' = (S',A',\delta')$ with $S' \subseteq S$,
$\emptyset \neq A'(s) \subseteq A(s)$ for all~$s \in S'$,
$\supp(\delta(s,a)) \subseteq
S'$ for all~$s \in S', a \in A'(s)$, $\delta' = \restr{\delta}{S'\times A'}$.
Such a sub-MDP $\mdp'$ is an \emph{end-component} (EC) of $\mdp$ if and only if the
underlying graph of $\mdp'$ is \textit{strongly connected}, i.e.,
there is a run between any pair of states in $S'$. Given an EC $\mdp' = (S',A',\delta')$ of $\mdp$, we say that its sub-MDP $\mdp'' = (S'',A'',\delta'')$, $S'' \subseteq S'$, $A'' \subseteq  A'$, is a \textit{sub-EC} of $\mdp'$ if $\mdp''$ is also an EC\@. We let $\ecs(\mdp)$
denote the set of ECs of~$\mdp$, which may be of exponential size as ECs need not be disjoint.

The union of two ECs with non-empty intersection is itself an EC:\@ hence we can define the \emph{maximal} ECs (MECs) of an MDP, i.e., the ECs that cannot be extended. We let $\mecs(\mdp)$
denote the set of MECs of~$\mdp$; it is of polynomial size (because MECs are pair-wise disjoints) and computable in polynomial
time~\cite{DBLP:journals/jacm/ChatterjeeH14}.

\paragraph{Strategies} A~\emph{strategy}~$\sigma$ is
a function $\Hists{\mdp}\rightarrow \dists(A)$ such that for
all~$h \in \Hists{\mdp}$ ending in~$s$, we
have~$\supp(\sigma(h)) \subseteq A(s)$. The set of all strategies is
$\strats$. A strategy is \textit{pure} if all histories are mapped to
\textit{Dirac distributions}, i.e., the support is a singleton. A
strategy~$\sigma$ can be encoded by a \emph{Mealy machine}
$(Q,\sigma_a,\sigma_u,\iota)$ where~$Q$ is a finite or infinite set
of memory states,~$\iota$ the \emph{initial distribution on~$Q$},
$\sigma_a$ the \emph{next action function}
$\sigma_a\colon S \times Q \rightarrow \dists(A)$ where
$\supp(\sigma_a(s,q)) \subseteq A(s)$ for any~$s\in S$ and~$q \in Q$,
and $\sigma_u$ the \emph{memory update function}
$\sigma_u\colon A\times S \times Q \rightarrow Q$.  We say
that~$\sigma$ is \emph{finite-memory} if~$|Q|<\infty$, and
\emph{$K$-memory} if~$|Q|=K$; it is memoryless if~$K=1$, thus only
depends on the last state of the history.  We see such strategies as
functions $s \mapsto \dists(A(s))$ for $s \in S$.  A strategy is
\emph{infinite-memory} if~$Q$ is infinite.
The entity choosing the strategy is often called the
\textit{controller}.

\paragraph{Induced MC} An MDP~$\mdp$, a strategy~$\sigma$ encoded by
$(Q,\sigma_a,\sigma_u,\iota)$, and a state~$s$ determine a Markov
chain $\mdp_s^\sigma$ defined on the state space $S\times Q$ as
follows.  The initial distribution is such that for any~$q \in Q$,
state $(s,q)$ has probability $\iota(q)$, and~$0$ for other
states. For any pair of states $(s,q)$ and~$(s',q')$, the probability
of transition $(s,q) \xrightarrow{a} (s',q')$ is equal to
$\sigma_a(s,q)(a) \cdot \delta(s,a,s')$ if $q' = \sigma_u(s,q,a)$, and
to~$0$ otherwise.  A~\emph{run} of~$\mdp_s^\sigma$ is an infinite
sequence of the form $(s_0,q_0)a_0(s_1,q_1)a_1\ldots$, where each
$(s_i,q_i) \xrightarrow{a_i} (s_{i+1},q_{i+1})$ is a transition with
non-zero probability in~$\mdp_s^\sigma$, and~$s_0=s$.  When
considering the probabilities of events in $\mdp_s^\sigma$, we will
often consider sets of runs of~$\mdp$. Thus, given
$E \subseteq {(SA)}^\omega$, we denote by $\pr_{\mdp,s}^\sigma[E]$ the
probability of the runs of~$\mdp_s^\sigma$ whose
projection\footnote{The projection of a run
  $(s_0,q_0)a_0(s_1,q_1)a_1\ldots$ in $\mdp_s^\sigma$ to $\mdp$ is
  simply the run $s_{0}a_{0}s_{1}a_{1}\ldots{}$ in $\mdp$. For the sake of readability, we make similar abuse of notation --- identifying runs in the induced MC with their projections in the MDP --- throughout our paper.} to~$\mdp$
is in~$E$, i.e., the probability of event $E$ when $\mdp$ is executed
with initial state~$s$ and strategy~$\sigma$. Note that every measurable set (\textit{event})
has a uniquely defined probability~\cite{vardi_FOCS85}
(Carath\'eodory's extension theorem induces a unique probability
measure on the Borel $\sigma$-algebra over cylinders of
${(SA)}^\omega$). We may drop some subscripts of $\pr_{\mdp,s}^\sigma$
when the context is clear.

\paragraph{Bottom strongly-connected components} The counterparts of ECs in MCs are \textit{bottom strongly-connected components} (BSCCs). In our formalism, where an MC is simply an MDP $\mdp = (S,A,\delta)$ with $\vert A(s)\vert = 1$ for all $s \in \states$, BSCCs are exactly the ECs of such an MDP $\mdp$.

\paragraph{Sure and almost sure events} Let $\mdp = (S,A,\delta)$, $\sigma \in \Sigma$, and $E \subseteq {(SA)}^\omega$
be an event. We say that
$E$ is \textit{sure}, written $\sureStrat{\mdp,s}{\sigma}{E}$, if and only if
  $\Runs{\mdp^\sigma_s} \subseteq E$ (again abusing our notation to consider projections on ${(SA)}^\omega$);
and that $E$ is \textit{almost-sure}, written
$\almostStrat{\mdp,s}{\sigma}{E}$, if and only if
  $\pr^\sigma_{\mdp,s}[E] = 1$.

\paragraph{Almost-sure reachability of ECs} Given a run $\rho = s_0 a_0 s_1 a_1 \ldots{} \in \Runs{\mdp}$, let
\[
\Inf(\rho) = \{s \in S \mid \forall\, i \geq 0,\, \exists\, j > i,\, s_j = s\}
\]
denote the set of states visited infinitely-often along $\rho$, and let
\[
\InfAct(\rho) = \{a \in A \mid \forall\, i \geq 0,\, \exists\, j > i,\, a_j = a\}
\]
similarly denote the actions taken infinitely-often along $\rho$.

Let $\InfRun(\rho)$ denote the pair $(\Inf(\rho), \InfAct(\rho))$. Note that this pair may induce a well-defined sub-MDP $\mdp' = (\Inf(\rho), \InfAct(\rho), \restr{\delta}{\Inf(\rho)\times \InfAct(\rho)})$, but in general this need not be the case. A folklore result in MDPs (e.g.,~\cite{baier2008principles}) is the following: for any state $s$ of MDP $\mdp$, for any strategy $\sigma \in \Sigma$, we have that
\[
\almostStrat{\mdp,s}{\sigma}{\{\rho \in \Runs{\mdp^\sigma_s} \mid \InfRun(\rho) \in \ecs(\mdp)\}},
\]
that is, under any strategy, the limit behavior of the MDP almost-surely coincides with an EC\@. This property is a key tool in the analysis of MDPs with prefix-independent objectives, as it essentially says that we only need to identify the ``best''  ECs and maximize the probability to reach them.

\paragraph{Decision problem} An \textit{objective} for an MDP $\mdp = (S, A, \delta)$ is a measurable set of runs $E \subseteq {(SA)}^\omega$. Given an MDP $\mdp=(S, A, \delta)$, an initial state $s$, a threshold $\alpha \in [0, 1] \cap \mathbb{Q}$, and such an objective $E$, the \textit{threshold probability problem} is to decide whether there exists a strategy $\sigma \in \Sigma$ such that $\mathbb{P}_{\mathcal{M}, s}^\sigma\left[ E \right ] \geq \alpha$ or not.

Furthermore, if it exists, we want to build such a strategy.

\paragraph{Weights and priorities} In this paper, we always assume an MDP $\mdp=(S, A, \delta)$ with either (i)~a \textit{weight} function $w\colon A \rightarrow \mathbb{Z}$ of largest absolute weight $W$, or (ii) a \textit{priority} function $p\colon S \rightarrow \{0, 1, \ldots{}, d\}$, with $d \leq \vert S\vert +1$ (w.l.o.g.). This choice is left implicit when the context is clear, to offer a unified view of mean-payoff and parity variants of window objectives.

\paragraph{Complexity} When studying the complexity of decision problems, we make the classical assumptions of the field: we consider the model size $\vert \mdp \vert$ to be polynomial in $\vert S \vert$ and the \textit{binary encoding} of weights and probabilities (e.g., $V = \log_2 W$, with $W$ the largest absolute weight), whereas we consider the largest priority $d$, as well as the upcoming window size $\lambda$, to be encoded in unary.
When a problem is polynomial in $W$, we say that it is \textit{pseudo}-polynomial: it would be polynomial if weights would be given in unary.

\paragraph{Mean-payoff and parity objectives}
We consider window objectives based on \textit{mean-payoff} and \textit{parity} objectives. Let us discuss those classical objectives.

The first one is a \textit{quantitative} objective, for which we consider \textit{weighted} MDPs. Let $\rho \in \Runs{\mdp}$ be a run of such an MDP\@.
The \textit{mean-payoff} of prefix $\rho[0, n]$ is
$\mathsf{MP}(\rho[0, n])=\frac{1}{n} \sum_{i=0}^{n-1} w(a_i)$, for $n>0$.
This is naturally extended to runs by considering the limit behavior.
The \textit{mean-payoff} of $\rho$ is $\mathsf{MP}(\rho) = \liminf_{n \rightarrow \infty} \mathsf{MP}(\rho[0, n])$. Given a threshold $\nu \in \mathbb{Q}$, the mean-payoff objective accepts all runs whose mean-payoff is above the threshold, i.e., $\mathsf{MeanPayoff}(\nu) = \{\rho \in \Runs{\mdp} \mid \mathsf{MP}(\rho) \geq \nu\}$. The corresponding threshold probability problem is in~$\mathsf{P}$ using linear programming, and pure memoryless strategies suffice  (see, e.g.,~\cite{DBLP:journals/fmsd/RandourRS17}). Note that $\nu$ can be taken equal to zero without loss of generality. Another variant of mean-payoff exists: it is defined using $\limsup$ instead of $\liminf$. In the classical one-dimension setting, optimal strategies for the two versions coincide and ensure the same thresholds (whereas the values may differ on an arbitrary run). In our setting, we will interpret the mean-payoff over a finite horizon, hence the two variants are de facto equivalent for our use.

The second objective, parity, is a \textit{qualitative} one, for which we consider MDPs with a priority function. The \emph{parity objective} requires that the smallest priority seen infinitely often along a run be even, i.e., $\mathsf{Parity} = \{\rho \in \Runs{\mdp} \mid \min_{s \in \Inf(\rho)} p(s) = 0 \pmod 2\}$. Again, the corresponding threshold probability problem is in $\mathsf{P}$ and pure memoryless strategies suffice~\cite{DBLP:conf/soda/ChatterjeeJH04}.

\clearpage
\section{Window objectives}%
\label{sec:window}
\subsection{Definitions}\label{subsec:defs}

~\paragraph{Good windows}
Given a weighted MDP $\mdp$ and $\lambda > 0$, we define the \textit{good window mean-payoff} objective
\begin{equation*}
\GW_{\mpSub}(\lambda) = \Big\{ \rho \in \Runs{\mdp} \mid \exists\, l < \lambda,\; \mathsf{MP}\big(\rho[0, l+1]\big) \geq 0\Big\}
\end{equation*}
requiring the existence of a window of size bounded by $\lambda$ and starting at the first position of the run, over which the mean-payoff is at least equal to zero (w.l.o.g.).

Similarly, given an MDP $\mdp$ with priority function $p$, we define the \textit{good window parity} objective,
\begin{equation*}
\GW_{\paritySub}(\lambda) = \Big\{ \rho \in \Runs{\mdp} \mid \exists\, l < \lambda,\; \big(p(\rho[l]) \bmod 2 = 0 \wedge \forall\, k < l,\; p(\rho[l]) < p(\rho[k])\big)\Big\}
\end{equation*}
requiring the existence of a window of size bounded by $\lambda$ and starting at the first position of the run, for which the last priority is even and is the smallest within the window.

To preserve our generic approach, we use subscripts $\mpSub$  and $\paritySub$ for mean-payoff and parity variants respectively.
So, given $\Omega = \{\mpSub, \paritySub\}$ and a run $\rho \in \Runs{\mdp}$,
we say that \textit{an $\Omega$-window is closed} in at most $\lambda$ steps from $\rho[i]$ if $\rho[i, \infty]$ is in $\GW_\Omega(\lambda)$. If a window is not yet closed, we call it \textit{open}.

\paragraph{Fixed variants} Given $\lambda > 0$, we define the \textit{direct fixed window} objective
\begin{equation*}
\DFW_{\Omega}(\lambda) = \left\{ \rho \in \Runs{\mdp} \mid \forall\, j \ge 0,\; \rho[j,\infty] \in \GW_{\Omega}(\lambda)\right\}
\end{equation*}
asking for all $\Omega$-windows to be closed within $\lambda$ steps along the run.

We also define the \textit{fixed window} objective
\begin{equation*}
\FW_{\Omega}(\lambda) = \left\{ \rho \in \Runs{\mdp} \mid \exists\, i \ge 0,\; \rho[i,\infty] \in \DFW_{\Omega}(\lambda)\right\}
\end{equation*}
that is the
\textit{prefix-independent} version of the previous one: it requires it to be eventually satisfied.

\paragraph{Bounded variant} Finally, we define the \textit{bounded window} objective
\begin{equation*}
\BW_{\Omega} = \left\{ \rho \in \Runs{\mdp} \mid \exists\, \lambda > 0,\; \rho \in \FW_{\Omega}(\lambda)\right\}
\end{equation*}
requiring the existence of a bound $\lambda$ for which the fixed window objective is satisfied. Note that this bound need not be \textit{uniform} along all runs in general. A \textit{direct} variant may also be defined, but turns out to be ill-suited in the stochastic context: we illustrate it in Sect.~\ref{subsec:illustration} and discuss its pitfalls in Sect.~\ref{sec:concl}. Hence we focus on the prefix-independent version in the following.

\subsection{Overview in games}\label{subsec:games}

Window mean-payoff and window parity objectives were considered in two-player zero-sum games~\cite{DBLP:journals/iandc/Chatterjee0RR15,DBLP:journals/corr/BruyereHR16}. The game setting is equivalent to deciding if there exists a controller strategy in an MDP such that the corresponding objective is \textit{surely} satisfied, i.e., by all consistent runs. We quickly summarize the main results here. In both cases, window objectives establish \textit{conservative approximations} of the classical ones, with \textit{improved complexity}.

For mean-payoff~\cite{DBLP:journals/iandc/Chatterjee0RR15}, (direct and prefix-independent) fixed window objectives can be solved in polynomial time, both in the model and the window size. Memory is in general needed for both players, but polynomial memory suffices. Bounded versions belong to $\mathsf{NP} \cap \mathsf{coNP}$ and are as hard as mean-payoff games. Memoryless strategies suffice for the controller but not for its opponent (infinite memory is needed for the prefix-independent case, polynomial memory suffices in the direct case). Interestingly, if the controller can win the bounded window objective, then a \textit{uniform bound} exists, i.e., there exists a window size $\lambda$ sufficiently large such that the bounded version coincides with the fixed one. Recall that this is not granted by definition. We will see that \textit{in MDPs, this does not hold} (Ex.~\ref{ex:uniform}). This uniform bound in games is however pseudo-polynomial.

For parity~\cite{DBLP:journals/corr/BruyereHR16}, similar results are obtained. The crucial difference is the uniform bound on $\lambda$, which also exists but in this case is equal to the number of states of the game. Thanks to that, all variants of window parity objectives belong to $\mathsf{P}$. The memory requirements are the same as for window mean-payoff.

\subsection{Illustration}%
\label{subsec:illustration}

\begin{exa}
We first go back to the example of Sect.~\ref{sec:intro}, depicted in Fig.~\ref{fig:intro_ex}.  Let $\mdp$ be this MDP\@. Fix run $\rho = {(s_1\, a\, s_2\, b\, s_3\, c)}^\omega$. We have that $\rho \not\in \FW_{\paritySub}(\lambda = 2)$ --- a fortiori, $\rho \not\in \DFW_{\paritySub}(\lambda = 2)$ --- as the window that opens in $s_1$ is not closed after two steps (because $s_1$ has odd priority $1$, and $2$ is not smaller than $1$ so does not suffice to answer it). If we now set $\lambda = 3$, we see that this window closes on time, as $0$ is encountered within three steps. As all other windows are immediately closed, we have $\rho \in \DFW_{\paritySub}(\lambda = 3)$ --- a fortiori, $\rho \in \FW_{\paritySub}(\lambda = 3)$ and $\rho \in \BW_\paritySub$.

Regarding the probability of these objectives, however, we have already argued that, for all $\lambda > 0$, $\pr_{\mdp, s_1}[\FW_\paritySub(\lambda)] = 0$, whereas $\pr_{\mdp, s_1}[\mathsf{Parity}] = 1$ since $s_3$ is almost-surely visited infinitely often but any time bound is almost-surely exceeded infinitely often too. Observe that $\BW_\paritySub = \bigcup_{\lambda > 0} \FW_{\paritySub}(\lambda)$, hence we also have that $\pr_{\mdp, s_1}[\BW_\paritySub] = 0$ (by countable additivity).

Similar reasoning holds for window mean-payoff objectives, by taking the weight function $w = \{a \mapsto -1, b \mapsto 0, c \mapsto 1\}$.\hfill$\triangleleft$

\end{exa}

With the next example, we illustrate one of the main differences between games and MDPs w.r.t.~window objectives. As discussed in Sect.~\ref{subsec:games}, in two-player zero-sum games, both for parity and mean-payoff, there exists a \textit{uniform bound} on the window size $\lambda$ such that the fixed variants coincide with the bounded ones. We show that this property does not carry over to  MDPs.

\begin{figure}[tbh]
  \begin{center}\begin{tikzpicture}[->,>=stealth',shorten >=1pt,auto,node
      distance=2.5cm,bend angle=45, scale=0.6, font=\normalsize]
      \tikzstyle{p1}=[draw,circle,text centered,minimum size=7mm,text width=4mm]
      \tikzstyle{p2}=[draw,rectangle,text centered,minimum size=7mm,text width=4mm]
      \tikzstyle{act}=[fill,circle,inner sep=1pt,minimum size=1.5pt,
      node distance=1cm]    \tikzstyle{empty}=[text centered, text
      width=15mm]
      \node[p1] (1) at (0,0) {$s$};
      \node[empty] at ($(1)+(0,1)$) {1};
      \node[p1] (2) at (5,0) {$t$};
      \node[empty] at ($(2)+(0,1)$) {0};
      \node[act] (1a) at (1.5,0) {};
      \node[empty] at ($(1a)+(0.3,0.3)$) {\scriptsize $a$};
      \node[act] (2a) at (6.5,0) {};
      \node[empty] at ($(2a)+(0,-0.4)$) {\scriptsize $b$};
      \node[]  (init) at (-1.6,0) {};
      \path[-latex']
      (init) edge (1)
      (1) edge (1a)
      (1a) edge node[above,xshift=0mm]{\scriptsize $0.5$} (2);
      \draw [->] (1a) to[out=90,in=45]
      node[above,xshift=0mm]{\scriptsize $0.5$} (1);
      \draw [->] (2) to[out=0,in=180] node[]{} (2a);
      \draw [->] (2a) to[out=90,in=60] node[above ,xshift=0mm]{\scriptsize $1$} (2);
    \end{tikzpicture}
    \caption{On this Markov chain, there is no uniform bound over all runs, in contrast to the game setting.}%
    \label{fig:exNoBound}
  \end{center}
\end{figure}
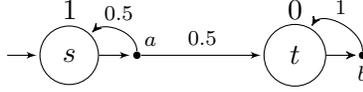

\begin{exa}%
\label{ex:uniform}

Consider the MC $\mdp$ depicted in Fig.~\ref{fig:exNoBound}. It is clear that for any $\lambda > 0$, there is probability $1/2^{\lambda-1}$ that objective $\DFW_{\paritySub}(\lambda)$ is not satisfied. Hence, for all $\lambda > 0$, $\pr_{\mdp, s}[\DFW_{\paritySub}(\lambda)] < 1$.

Now let $\DBW_{\paritySub}$ be the \textit{direct} bounded window objective evoked in Sect.~\ref{subsec:defs} and defined as $\DBW_\paritySub = \bigcup_{\lambda > 0} \DFW_{\paritySub}(\lambda)$. We claim that $\pr_{\mdp, s}[\DBW_{\paritySub}] = 1$. Indeed, any run ending in $t$ belongs to $\DBW_{\paritySub}$, as it belongs to $\DFW_{\paritySub}(\lambda)$ for $\lambda$ equal to the length of the prefix up to $t$. Since $t$ is almost-surely reached (as it is the only BSCC of the MC), we conclude that $\DBW_{\paritySub}$ is indeed satisfied almost-surely.

Essentially, the difference stems from the fairness of probabilities. In a game, the opponent would control the successor choice after action $a$ and always go back to $s$, resulting in objective $\DBW_{\paritySub}$ being lost. However, in this MC, $s$ is almost-surely left eventually, but we cannot guarantee when: hence there exists a window bound for each run, but there is no uniform bound over all runs.

This simple MC also illustrates the classical difference between sure and almost-sure satisfaction: while we have $\almostStrat{\mdp,s}{}{\FW_\paritySub(\lambda = 1)}$, we do not have $\sureStrat{\mdp,s}{}{\FW_\paritySub(\lambda = 1)}$, because of run $\rho = {(s\,a)}^\omega$.

Again the same reasoning holds for mean-payoff variants of the objectives, for example with $w = \{ a \mapsto -1, b \mapsto 1\}$.\hfill$\triangleleft$
\end{exa}

We leave out the \textit{direct} bounded objective $\DBW_\Omega$ from now on. We will come back to it in Sect.~\ref{sec:concl} and motivate why this objective is not well-behaved. Hence, in the following, we focus on direct and prefix-independent fixed window objectives and prefix-independent bounded ones.

\section{Fixed case: better safe than sorry}%
\label{sec:fixed}
We start with the fixed variants of window objectives. Our main goal here is to establish that \textit{pure finite-memory} strategies suffice in all cases. As a by-product, we also obtain algorithms to solve the corresponding decision problems. Still, for the prefix-independent variants, we will obtain better complexities using the upcoming generic approach (Sect.~\ref{sec:general}).

Our main tools are natural reductions from direct (resp.~prefix-independent) window problems on MDPs to safety (resp.~co-B\"uchi) problems on \textit{unfoldings} based on the window size $\lambda$ (i.e., larger arenas incorporating information on open windows). We use \textit{identical unfoldings} for both direct and prefix-independent objectives, in order to obtain a unified proof. In both cases, the set of states to avoid in the unfolding corresponds to leaving a window open for $\lambda$ steps.

\subsection{Reductions}\label{subsec:reductions}

~\paragraph{Mean-payoff}
Let $\mdp = (S,A,\delta)$ be an MDP with weight function $w$ (of maximal absolute weight $W$), and
$\lambda>0$ be the window size. We
define the unfolding MDP
$\lmdp = (\tS, A, \tdelta)$ as follows:
\begin{itemize}
\item $\tS = S \times  \set{0,\ldots, \lambda} \times \set{ -\lambda \cdot W, \ldots, 0}$.
\item $\tdelta\colon \tS \times A \to
  \dists(\tS)$ is defined as follows for all $a$ in $A$:
    \begin{align*}
      \tdelta\left((s,l,z), a\right)(t, l + 1, z + w(a)) = \nu
      \text{ if }
        \big( \delta(s,a)(t) = \nu \big) \; \wedge \;
        \big( l < \lambda          \big) \; \wedge \;
        \big( z + w(a) < 0         \big),
      \\
      \tdelta\left((s,l,z), a\right)(t,0,0) = \nu
      \text{ if }
        \big( \delta(s,a)(t) = \nu \big) \; \wedge \; \Big[\big( z + w(a) \ge 0       \big) \vee \Big(\big ( l  = \lambda \big) \; \wedge \;
        \big ( z < 0 \big)\Big)\Big].
    \end{align*}
  \item Once an initial state $s_\text{init} \in S$ is fixed in $\mdp$, the
    associated initial state in $\lmdp$ is $\tilde{s}_\text{init} = (s_\text{init}, 0, 0) \in \tilde{S}$.
  \end{itemize}

  \noindent
  Note that $\lmdp$ is unweighted: it keeps track in each of its
  states of the current state of $\mdp$, the size of the
  current open window as well as the current sum of weights in the
  window: these two values are reset whenever a window is closed (left-hand side of the disjunction) or stays open for $\lambda$ steps (right hand-side). A key underlying property used here is the so-called
  \textit{inductive property of
    windows}~\cite{DBLP:journals/iandc/Chatterjee0RR15,DBLP:journals/corr/BruyereHR16}.

\begin{pty}[Inductive property of windows]
  Consider a run $\rho = s_0 a_0 s_1 a_1 \ldots{}$ in an MDP\@. Fix a
  window starting in position $i \geq 0$. Let $j$ be the position in
  which this window gets closed, assuming it does. Then, all windows
  in positions from $i$ to $j$ also close in $j$.
\end{pty}
The validity of this property is easy to check by contradiction (if it
was not the case, then the window in $i$ would close before $j$). This
property is fundamental in our reduction: without it we would have to
keep track of all open windows in parallel, which would result in a
blow-up exponential in $\lambda$.

In the upcoming reductions, the set of states to avoid will be \[B = \set{(s,l,z) \mid (l = \lambda) \land (z < 0)}.\] By construction of $\lmdp$, it exactly corresponds to windows staying open for $\lambda$ steps.

\paragraph{Parity}  Let $\mdp = (S,A,\delta)$ be an
MDP with priority function $p$, and $\lambda>0$ be the window size. We define the unfolding MDP
$\lmdp = (\tS, A, \tdelta)$ as follows:
\begin{itemize}
\item $\tS = S \times  \set{0,\ldots, \lambda} \times \set{0, 1, \dots, d}$.
\item $\tdelta\colon \tS \times A \to
  \dists(\tS)$ is defined as follows for all $a$ in $A$:
    \begin{align*}
      \tdelta\left((s,l,c), a\right)(t, l + 1, \min(c, p(t))) &= \nu
      \text{ if }
        \big( \delta(s,a)(t) = \nu \big) \; \wedge \;
        \big( l < \lambda - 1      \big) \; \wedge \;
        \big( c \bmod 2 = 1         \big),
      \\
      \tdelta\left((s,l,c), a\right)(t,0,p(t)) &= \nu
      \text{ if }
        \big( \delta(s,a)(t) = \nu \big) \; \wedge \;
        \big( l  = \lambda - 1  \big) \; \wedge \;
        \big( c \bmod 2 = 1      \big),
        \\
      \tdelta\left((s,l,c), a\right)(t,0,p(t)) &= \nu
      \text{ if }
        \big( \delta(s,a)(t) = \nu \big) \; \wedge \;
        \big( c \bmod 2 = 0         \big).
    \end{align*}
  \item Once an initial state $s_\text{init} \in S$ is fixed in $\mdp$, the
    associated one in $\lmdp$ is $\tilde{s}_\text{init} = (s_\text{init}, 0, p(s_\text{init})) \in \tilde{S}$.
  \end{itemize}
  This unfolding keeps track in each of its states of the current state
of the original MDP, the size of the current window and the minimum
priority in the window: again, these two values are reset whenever a window is closed or stays open for $\lambda$ steps.

In this case, the set of states to avoid will be $B = \set{(s,l,c) \mid (l = \lambda - 1) \land (c \bmod 2 = 1)}$, with an equivalent interpretation.

\begin{rem}
There is a slight asymmetry in the use of indices for the current
window size in the two constructions: that is due to mean-payoff using
weights on actions (two states being needed for one action) and parity
using priorities on states.
\end{rem}

\paragraph{Objectives} We treat mean-payoff and parity versions in a unified way from now on. As stated before, the set $B$ represents in both cases windows being open for $\lambda$ steps and should be avoided: at all times for direct variants, eventually for prefix-independent ones. We define the following objectives over the unfoldings:
\begin{align*}
  &\reach(\lmdp) = {(\tilde{S}A)}^*BA{(\tilde{S}A)}^\omega, &\safety(\lmdp) = {(\tilde{S}A)}^\omega \setminus \reach(\lmdp),\\
  &\buchi(\lmdp) = {({(\tilde{S}A)}^*BA)}^\omega,  &\cobuchi(\lmdp) = {(\tilde{S}A)}^\omega \setminus \buchi(\lmdp).
\end{align*}
Our ultimate goal is to prove that the safety and co-B\"uchi objectives in $\lmdp$ are \textit{probability-wise equivalent} to the direct fixed window and fixed window ones in $\mdp$, modulo a well-defined \textit{mapping} between histories, runs and strategies. This will induce the correctness of our reduction.

\paragraph{Mapping} Fix an initial state $s_\text{init}$ in $\mdp$ and let $\tilde{s}_\text{init}$ be its corresponding initial state in $\lmdp$. Let $\Hists{\mdp, s_\text{init}}$ denote the histories of $\mdp$ starting in $s_\text{init}$. We start by defining a bijective mapping $\pi_{\lmdp}$, and its inverse $\pi_\mdp$, between \textit{histories} of $\Hists{\mdp, s_\text{init}}$ and $\Hists{\lmdp, \tilde{s}_\text{init}}$. We use $\pi_{\lmdp} \colon \Hists{\mdp, s_\text{init}} \rightarrow \Hists{\lmdp, \tilde{s}_\text{init}}$ for the $\mdp$-to-$\lmdp$ direction, and $\pi_\mdp\colon \Hists{\lmdp, \tilde{s}_\text{init}} \rightarrow \Hists{\mdp, s_\text{init}}$ for the opposite one. We define $\pi_{\lmdp}$ inductively as follows.
\begin{itemize}
\item $\pi_{\lmdp}(s_\text{init}) = \tilde{s}_\text{init}$.
\item Let $h \in \Hists{\mdp, s_\text{init}}$, $\tilde{h} = \pi_{\lmdp}(h)$, $a \in A$, $s \in S$. Then, $\pi_{\lmdp}(h\cdot a \cdot s) = \tilde{h} \cdot a \cdot \tilde{s}$, where $\tilde{s}$ is obtained from $\mathsf{Last}(\tilde{h})$, $a$ and $s$ following the unfolding construction.
\end{itemize}
We define $\pi_\mdp$ as its inverse, i.e., the function projecting histories of $\Hists{\lmdp, \tilde{s}_\text{init}}$ to ${(SA)}^*S$.
We naturally extend these mappings to \textit{runs} based on this inductive construction.

Now, we also extend these mappings to \textit{strategies}. Let $\sigma$ be a strategy in $\mdp$. We define its twin strategy $\tilde{\sigma} = \pi_{\lmdp}(\sigma)$ in $\lmdp$ as follows: for all $\tilde{h} \in \Hists{\lmdp, \tilde{s}_\text{init}}$, $\tilde{\sigma}(\tilde{h}) = \sigma(\pi_\mdp(\tilde{h}))$. Note that this strategy is well-defined as $\mdp$ and $\lmdp$ share the same actions and $\pi_\mdp$ is a proper function over $\Hists{\lmdp, \tilde{s}_\text{init}}$ (i.e., each history $\tilde{h}$ has an image in $\Hists{\mdp, s_\text{init}}$). Similarly, given a strategy $\tilde{\sigma}$ in $\lmdp$, we build a twin strategy $\sigma = \pi_\mdp(\tilde{\sigma})$ using $\pi_{\lmdp}$. Hence we also have a bijection over strategies.

We say that two objects (histories, runs, strategies) are $\pi$-\textit{corresponding} if they are the image of one another through mappings $\pi_\mdp$ and $\pi_{\lmdp}$.

\paragraph{Probability-wise equivalence} For any any history $h$, we denote by $\cyl(h)$ the \textit{cylinder set} spanned by it, i.e., the set of all runs prolonging $h$. Cylinder sets are the building blocks of probability measures in MCs, as all measurable sets belong to the $\sigma$-algebra built upon them~\cite{baier2008principles}. We show that our mappings preserve probabilities of cylinder sets.

\begin{lem}%
\label{lem:cylinders}
Let $\mdp$, $\lmdp$, $\pi_{\mdp}$ and $\pi_{\lmdp}$ be defined as above. Fix any couple of $\pi$-corresponding strategies $(\sigma, \tilde{\sigma})$ for $\mdp$ and $\lmdp$ respectively. Then, for any couple of $\pi$-corresponding histories $(h, \tilde{h})$ in $\Hists{\mdp, s_\text{init}} \times \Hists{\lmdp, \tilde{s}_\text{init}}$, we have that
\begin{equation}
\label{eq:cylinders}
\pr_{\mdp, s_\text{init}}^\sigma[\cyl(h)] = \pr_{\lmdp, \tilde{s}_\text{init}}^{\tilde{\sigma}}[\cyl(\tilde{h})].
\end{equation}
\end{lem}
\begin{proof}
We fix $\mdp$, $\lmdp$, $\pi_{\mdp}$, $\pi_{\lmdp}$, and a couple of $\pi$-corresponding strategies $(\sigma, \tilde{\sigma})$. We prove the equality by induction on histories. The base case, for $h = s_\text{init}$ and $\tilde{h} = \tilde{s}_\text{init}$, is trivial, as we have
\begin{align*}
\pr_{\mdp, s_\text{init}}^\sigma[\cyl(h)] = \pr_{\mdp, s_\text{init}}^\sigma[\Runs{\mdp, s_\text{init}}]  &= 1\\ &=  \pr_{\lmdp, \tilde{s}_\text{init}}^{\tilde{\sigma}}[\Runs{\lmdp, \tilde{s}_\text{init}}] = \pr_{\lmdp, \tilde{s}_\text{init}}^{\tilde{\sigma}}[\cyl(\tilde{h})].
\end{align*}
Now, assume Eq.~\eqref{eq:cylinders} true for a couple of $\pi$-corresponding histories $(h, \tilde{h})$. Consider any one-step extension of these histories, say $(h\cdot a \cdot s, \tilde{h}\cdot a\cdot \tilde{s})$ for $a \in A$, $s \in S$ and $\tilde{h}\cdot a\cdot \tilde{s} = \pi_{\lmdp}(h\cdot a\cdot s)$. We need to prove that Eq.~\eqref{eq:cylinders} still holds for this pair of histories. Let us expand the equality to prove as follows:
\begin{align*}
\pr_{\mdp, s_\text{init}}^\sigma[\cyl(h\cdot a \cdot s)] &= \pr_{\lmdp, \tilde{s}_\text{init}}^{\tilde{\sigma}}[\cyl(\tilde{h}\cdot a \cdot \tilde{s})]\\
\iff \pr_{\mdp, s_\text{init}}^\sigma[h] \cdot \sigma(h)(a) \cdot \delta(\mathsf{Last}(h), a, s) &= \pr_{\lmdp, \tilde{s}_\text{init}}^{\tilde{\sigma}}[\tilde{h}] \cdot \tilde{\sigma}(\tilde{h})(a) \cdot \tilde{\delta}(\mathsf{Last}(\tilde{h}), a, \tilde{s}).
\end{align*}
Now, observe that $\sigma(h)(a) = \tilde{\sigma}(\tilde{h})(a)$ since $\sigma$ and $\tilde{\sigma}$ are $\pi$-corresponding strategies and $h$ and $\tilde{h}$ are $\pi$-corresponding histories. Furthermore, $\delta(\mathsf{Last}(h), a, s) = \tilde{\delta}(\mathsf{Last}(\tilde{h}), a, \tilde{s})$ by construction of the unfolding. Now since $\pr_{\mdp, s_\text{init}}^\sigma[h] = \pr_{\lmdp, \tilde{s}_\text{init}}^{\tilde{\sigma}}[\tilde{h}]$ holds by induction hypothesis, we are done.
\end{proof}

Since the previous lemma holds for all cylinders, we may extend the result to any event (see, e.g.,~\cite{baier2008principles} for the construction of the $\sigma$-algebra and related questions).

\begin{cor}%
\label{cor:events}
Let $\mdp$, $\lmdp$, $\pi_{\mdp}$ and $\pi_{\lmdp}$ be defined as above. Fix any couple of $\pi$-corresponding strategies $(\sigma, \tilde{\sigma})$ for $\mdp$ and $\lmdp$ respectively. Then, for any couple of $\pi$-corresponding measurable objectives $(E, \tilde{E}) \subseteq \Runs{\mdp, s_\text{init}} \times \Runs{\lmdp, \tilde{s}_\text{init}}$, we have that
\begin{equation*}
\pr_{\mdp, s_{\text{init}}}^\sigma[E] = \pr_{\lmdp, \tilde{s}_\text{init}}^{\tilde{\sigma}}[\tilde{E}].
\end{equation*}
\end{cor}

\paragraph{Correctness of the reductions} We may now establish our reductions.

\begin{lem}%
  \label{lem:reductions}
  Let $\mdp = (S,A,\delta)$ be an MDP, $\lambda > 0$ be the window size, $\Omega
  \in \{\mpSub, \paritySub\}$, $\mdp_\lambda = (\tS, A, \tdelta)$ be the
  unfolding of $\mdp$ defined as above for variant $\Omega$, $s \in S$ be a state of $\mdp$, and $\tilde{s} \in \tilde{S}$ be its
 $\pi$-corresponding state in $\mdp_\lambda$. The following assertions hold.
  \begin{enumerate}
  \item For any strategy $\sigma$ in $\mdp$, there exists a strategy
      $\tilde{\sigma}$ in $\lmdp$ such that
  \begin{align*}
    \pr_{\lmdp, \tilde{s}}^{\tilde{\sigma}}[\safety(\lmdp)] = \pr_{\mdp, s}^\sigma[\DFW_{\Omega}(\lambda)] \quad \wedge \quad \pr_{\lmdp, \tilde{s}}^{\tilde{\sigma}}[\cobuchi(\lmdp)] = \pr_{\mdp, s}^\sigma[\FW_{\Omega}(\lambda)].
  \end{align*}
  \item For any strategy $\tilde{\sigma}$ in $\lmdp$, there exists a strategy
      $\sigma$ in $\mdp$ such that
  \begin{align*}
    \pr_{\mdp,s}^\sigma[\DFW_{\Omega}(\lambda)] =
    \pr_{\lmdp, \tilde{s}}^{\tilde{\sigma}}[\safety(\lmdp)] \quad \wedge \quad \pr_{\mdp,s}^\sigma[\FW_{\Omega}(\lambda)] =
    \pr_{\lmdp, \tilde{s}}^{\tilde{\sigma}}[\cobuchi(\lmdp)].
  \end{align*}
  \end{enumerate}
  Moreover, such strategies can be obtained through mappings $\pi_{\mdp}$ and $\pi_{\lmdp}$.
\end{lem}
\begin{proof}
To prove it, it suffices to see that $\safety(\lmdp)$ (resp.~$\cobuchi(\lmdp)$) is $\pi$-corresponding with $\DFW_{\Omega}(\lambda)$ (resp.~$\FW_{\Omega}(\lambda)$) and invoke Cor.~\ref{cor:events}. As noted earlier, this correspondence is trivial by construction of the unfoldings. Consider the safety case: a run $\tilde{\rho}$ in $\lmdp$ belongs to $\safety(\lmdp)$  \textit{if and only if} all the windows along $\rho = \pi_\mdp(\tilde{\rho})$ close within $\lambda$ steps. Similarly, a run $\tilde{\rho}$ in $\lmdp$ belongs to $\cobuchi(\lmdp)$  \textit{if and only if} it visits $B$ finitely often, hence if and only if $\rho = \pi_\mdp(\tilde{\rho})$ contains a finite number of windows left open for $\lambda$ steps.
\end{proof}

Intuitively, to obtain a strategy $\sigma$ in $\mdp$ from a strategy $\tilde{\sigma}$ in the unfolding $\lmdp$, we have to integrate in the memory of $\sigma$ the additional information encoded in $\tilde{S}$: hence the memory required by $\sigma$ is the one used by $\tilde{\sigma}$ with a blow-up polynomial in $\vert \tilde{S}\vert$.

\subsection{Memory requirements and complexity}\hfill

\paragraph{Upper bounds} Thanks to the reductions established in Lem.~\ref{lem:reductions}, along with the
  fact that \textit{pure memoryless} strategies suffice for safety and co-Büchi
  objectives in MDPs~\cite{baier2008principles}, we obtain the following result.

  \begin{thm}\label{finite-memory-cor}
Pure finite-memory strategies suffice for the threshold probability problem for all fixed window objectives.
That is,
given MDP $\mdp = (S,A,\delta)$, initial state $s \in S$, window size $\lambda > 0$, $\Omega  \in \{\mpSub, \paritySub\}$, objective $E \in \{\DFW_{\Omega}(\lambda), \FW_{\Omega}(\lambda)\}$ and threshold probability $\alpha \in [0, 1] \cap \mathbb{Q}$, if there exists a strategy $\sigma \in \Sigma$ such that $\pr_{\mdp,s}^\sigma [E] \geq \alpha$, then there exists a pure finite-memory strategy $\sigma'$ such that  $\pr_{\mdp,s}^{\sigma'} [E] \geq \alpha$.
  \end{thm}

These reductions also yield algorithms for the threshold probability problem in the fixed window case. We only use them for the direct variants, as the generic approach we develop in Sect.~\ref{sec:general} proves to be more efficient for the prefix-independent one, for two reasons: first, we may restrict the co-B\"uchi-like analysis to end-components; second, we use a more tractable analysis than the co-B\"uchi unfolding for mean-payoff. However, the reduction established for prefix-independent variants is not without interest: it yields sufficiency of finite-memory strategies, which is a \textit{key ingredient} to establish our generic approach (used in Lem.~\ref{grolem}).

\begin{thm}%
\label{thm:direct}
  The threshold probability problem is
\begin{enumerate}[label={(\alph*)}]
\item in $\mathsf{P}$ for direct fixed window parity objectives, and pure polynomial-memory optimal strategies can be constructed in polynomial time.
\item in $\mathsf{EXPTIME}$ for direct fixed window mean-payoff objectives, and pure pseudo-polynomial-memory optimal strategies can be constructed in pseudo-polynomial time.
\end{enumerate}
\end{thm}

\begin{proof}
The algorithm is simple: given $\mdp = (S,A,\delta)$ and $\lambda > 0$, build $\lmdp$ and solve the corresponding safety problem. This can be done in polynomial time in $\vert\lmdp\vert$ and pure memoryless strategies suffice over $\lmdp$~\cite{baier2008principles}.

For parity, the unfolding is of size polynomial in $\vert \mdp \vert$, the number of priorities $d$ and the window size $\lambda$. Since both $d$ (anyway bounded by $\mathcal{O}(\vert S\vert)$) and $\lambda$ are assumed to be given in unary, it yields the result.

For mean-payoff, the unfolding is of size polynomial in $\vert \mdp \vert$, the largest absolute weight $W$ and the window size $\lambda$. Since weights are assumed to be encoded in binary, we only have a pseudo-polynomial-time algorithm.
\end{proof}

\paragraph{Lower bounds} We complement the results of Thm.~\ref{thm:direct} with almost-matching lower bounds, showing that our approach is close to optimal, complexity-wise.

\begin{thm}%
\label{thm:directHardness}
The threshold probability problem is
\begin{enumerate}[label={(\alph*)}]
\item\label{item:directHardness1} $\PTIME$-hard for direct fixed window parity objectives, and polynomial-memory strategies are in general necessary;
\item\label{item:directHardness2} $\mathsf{PSPACE}$-hard for direct fixed window mean-payoff objectives (even for acyclic MDPs), and pseudo-polynomial-memory strategies are in general necessary.
\end{enumerate}
\end{thm}
\begin{proof}
Item~\ref{item:directHardness1}. We establish $\PTIME$-hardness for direct fixed window parity objectives through a reduction from two-player reachability games, which are known to be $\PTIME$-complete~\cite{DBLP:journals/tods/Beeri80,DBLP:journals/jcss/Immerman81}.  Let $\mathcal{G} = (V = V_1 \uplus V_2, E)$ be a game graph, where states in $V_1$ (resp.~$V_2$) belong to player~$1$ (resp.~player~$2$), and $E \subseteq V \times V$ is the set of transitions. Without loss of generality, we assume this graph to be strictly alternating, i.e., $E \subseteq V_1 \times V_2 \uplus V_2 \times V_1$, and without deadlock. The reachability objective $\mathsf{Reach}(T)$ for $T \subseteq V$ accepts all plays that eventually visit the set $T$. Again w.l.o.g., we assume that $T \subseteq V_1$. Given an initial state $v_{\text{init}} \in V_1$ and target set $T \subseteq V_1$, deciding if player~$1$ has a winning strategy ensuring that $T$ is visited whatever the strategy of player~$2$ is $\PTIME$-hard. We reduce this question to a threshold probability problem for a direct fixed window parity objective as follows. From $\mathcal{G}$, we build the MDP $\mdp = (S, A, \delta)$ such that $S = V_1$, $A = V_2$ and $\delta$ is constructed in the following manner:
\begin{itemize}
\item for all $v_1\in V_1 \setminus T$, $v'_1 \in V_1$, $v_2 \in V_2$, $(v_1, v_2, v'_1) \in \delta$ iff $(v_1, v_2) \in E$ and $(v_2, v'_1) \in E$;
\item probabilities are taken uniform, i.e., $\delta(v_1, v_2, v'_1) = \left( \dfrac{1}{\vert\supp(\delta(v_1, v_2))\vert}\right) $;
\item states in $T$ are made absorbing, i.e., they only allow an action $a$ such that $\delta(v, a, v) = 1$ for any $v \in T$.
\end{itemize}
We add a priority function $p\colon S \rightarrow \{0, 1\}$ that assigns $1$ to all states in $S$ except states that correspond to states in $T \subseteq V_1$: those states get priority $0$. Let us fix objective $\DFW_\paritySub(\lambda = \vert V_1 \vert)$. We claim that player~$1$ has a winning strategy from $v_{\text{init}}$ in the reachability game $\mathcal{G}$ if and only if the controller has a strategy almost-surely satisfying objective $\DFW_\paritySub(\lambda = \vert V_1 \vert)$ from $v_{\text{init}}$ in the MDP $\mdp$. Observe that this reduction is in $\LOGSPACE$: it remains to check its correctness.

Consider what happens in $\mdp$. The only way to close the window that will open in $v_{\text{init}}$ is to reach  $T$, and we ought to close it to obtain a satisfying run as we consider the \textit{direct} objective. Furthermore, once $T$ is reached, the run is necessarily winning as we stay in $T$ forever and always see priority $0$. We also know that if $T$ can be reached, $\lambda$ steps suffice to do so, as memoryless strategies suffice in the game $\mathcal{G}$ (note that we do not count actions here). Lifting a winning strategy from $\mathcal{G}$ to $\mdp$ is thus trivial: we mimic a memoryless one w.l.o.g.\ and reach $T$ \textit{surely} in $\mdp$ within $\lambda$ steps, ensuring the objective.

Hence, the other direction remains: given a strategy $\sigma$ such that $\almostStrat{\mdp, v_{\text{init}}}{\sigma}{\DFW_\paritySub(\lambda)}$, we need to construct a winning strategy in $\mathcal{G}$. This may seem more difficult as we need to go from an almost-surely winning strategy to a surely winning one: something which is not possible in general. Yet, we show that $\sigma$ is actually surely winning for $\DFW_\paritySub(\lambda)$. By contradiction, assume it is not the case, i.e., that there exists a consistent run $\rho$ such that $\rho \not\in \DFW_\paritySub(\lambda)$: then, a finite prefix $\rho[0, n]$ can be extracted, such that a window remains open for $\lambda$ steps along it. Now, since we consider a direct objective, the cylinder spanned by this prefix only contains losing runs, and since $\rho[0, n]$ is finite, it has a strictly positive probability. Hence $\sigma$ is not almost-surely winning and we have our contradiction. This proves that $\sigma$ is actually surely-winning, and it is then trivial to mimic it in $\mathcal{G}$ to obtain a winning strategy for player~$1$. This shows that the reduction from reachability games to the threshold probability problem for direct fixed window parity objectives holds, thus that the latter is $\PTIME$-hard.

Regarding memory, the proof for direct fixed window parity games in~\cite{DBLP:journals/corr/BruyereHR16} carries over easily to our setting by replacing the states of the opponent by stochastic actions, in the natural way. Hence the lower bound is trivial to establish. Yet, we illustrate the need for memory in Ex.~\ref{ex:directmem} to help the reader understand the phenomenon.

Item~\ref{item:directHardness2}. Consider the $\mathsf{PSPACE}$-hardness of the mean-payoff variant. We proceed via a reduction from the threshold probability problem for shortest path objectives~\cite{DBLP:conf/icalp/HaaseK15,DBLP:journals/fmsd/RandourRS17}. This problem is as follows. Let $\mdp = (S, A, \delta)$ be an MDP with weight function $w\colon A \rightarrow \mathbb{N}_0$ (we use strictly positive weights w.l.o.g.). We fix a target set $T \subseteq S$ and define the \textit{truncated sum up to} $T$ as the function $\mathsf{TS}^T\colon \Runs{\mdp} \rightarrow \mathbb{N} \cup \{\infty\}$ given by
\begin{equation*}
\mathsf{TS}^T(\rho) = \begin{cases}
\sum_{i = 0}^{n-1} w(a_i) &\text{if } \rho[n] \in T \,\wedge\, \forall\, 0 \leq j < n,\, \rho[j] \not\in T,\\
\infty &\text{if } \forall\, j \geq 0,\, \rho[j] \not\in T.\end{cases}
\end{equation*}
Given an upper bound $\ell \in \mathbb{N}$, the shortest path objective is \[\mathsf{ShortestPath}(\ell) = \{\rho \in \Runs{\mdp} \mid \mathsf{TS}^T(\rho) \leq \ell\}.\] Deciding if there exists a strategy $\sigma$ such that  $\pr_{\mdp,s}^\sigma[\mathsf{ShortestPath}(\ell)] \geq \alpha$, given $s \in S$ and $\alpha \in [0, 1] \cap \mathbb{Q}$, is known to be $\mathsf{PSPACE}$-hard, even for \textit{acyclic} MDPs~\cite{DBLP:conf/icalp/HaaseK15}. The target set $T$ is assumed to be made of absorbing states (i.e., with self-loops): the acyclicity is to be interpreted over the rest of the underlying graph. We establish a reduction from this problem, in the acyclic case, to a threshold probability problem for a direct fixed window mean-payoff objective, maintaining the acyclicity of the underlying graph (except in target states, again).

Given the original MDP $\mdp$, we only modify the weight function to obtain MDP $\mdp'$. We define $w'$ from $w$, the target $T$ and the bound $\ell$ as follows:
\begin{itemize}
\item for all $(s, a, s') \in \delta$, $s' \not\in T$, $w'(a) = -w(a)$;
\item for all $(s, a, t) \in \delta$, $s \not\in T$, $t \in T$, $w'(a) = -w(a) + \ell$;
\item for all $(t, a, t) \in \delta$, $t \in T$, $w'(a) = 0$.
\end{itemize}
Intuitively, we take the opposite of all weights; add the bound when entering the target; and make the target cost-free. We then define the objective $\DFW_\mpSub(\lambda = \vert \states \vert)$ and claim that there exists a strategy $\sigma$ in $\mdp$ to ensure $\pr_{\mdp,s}^\sigma[\mathsf{ShortestPath}(\ell)] \geq \alpha$ if and only if there exists a strategy $\sigma'$ in $\mdp'$ to ensure $\pr_{\mdp',s}^{\sigma'}[\DFW_\mpSub(\lambda)] \geq \alpha$.

Proving it is fairly easy. Observe that, by construction, the sum of weights over a prefix in $\mdp'$ that is not yet in $T$ is strictly negative, and the opposite of the sum over the same prefix in the original MDP $\mdp$. Due to the addition of $\ell$ on entering $T$, we have that any run $\rho' \in \Runs{\mdp'}$ sees all its windows closed if and only if the very same run $\rho \in \Runs{\mdp}$ is such that $\mathsf{TS}^T(\rho) \leq \ell$ in $\mdp$. Now, using the acyclicity of the underlying graph, we know that if a run reaches $T$, it does so in at most $\lambda$ steps. We thus conclude that $\rho' \in \DFW_\mpSub(\lambda)$ in $\mdp'$ if and only if $\rho \in \mathsf{ShortestPath}(\ell)$ in $\mdp$. It is then trivial to derive the desired result and establish the correctness of the reduction: this concludes our proof of $\mathsf{PSPACE}$-hardness for the threshold probability problem for direct fixed window mean-payoff objectives.

Finally, the need for pseudo-polynomial memory is also obtained through this reduction. Indeed, there is a chain of reductions from \textit{subset-sum games}~\cite{DBLP:journals/tcs/Travers06,DBLP:journals/iandc/FearnleyJ15} to our setting, via the shortest path problem presented above~\cite{DBLP:conf/icalp/HaaseK15}. More precisely, subset-sum games are known to require pseudo-polynomial-memory strategies as optimal strategies need to track the current sum of weights (see, e.g.,~\cite{DBLP:journals/fmsd/RandourRS17}). These games can be (polynomially) reduced to the shortest path problem, as in~\cite{DBLP:conf/icalp/HaaseK15}, which in turn we can reduce (polynomially) to our problem, as just established. Hence the need for pseudo-polynomial-memory strategies carries over to our setting.
\end{proof}

\begin{rem}
Throughout this paper, we assume the window size $\lambda$ to be given in unary, or equivalently, to be polynomial in the description of the MDP (as for practical purposes, having exponential window sizes would somewhat defeat the purpose of time bounds in specifications). It is thus interesting to note that the hardness results we established in Thm.~\ref{thm:directHardness} do hold under this assumption: the complexity essentially comes from the weight structure, as in other counting-like problems in MDPs~\cite{DBLP:conf/icalp/HaaseK15,DBLP:journals/fmsd/RandourRS17,DBLP:journals/iandc/BruyereFRR17}.
\end{rem}

\begin{rem}%
\label{rmk:DFWMP_AS_P}
As noted in the proof of Thm.~\ref{thm:directHardness}, almost-surely winning coincides with surely winning for the \textit{direct fixed} window objectives. Therefore, the threshold probability problem for $\DFW_\mpSub(\lambda)$ collapses to $\mathsf{P}$ if the threshold is $\alpha = 1$~\cite{DBLP:journals/iandc/Chatterjee0RR15}.
\end{rem}

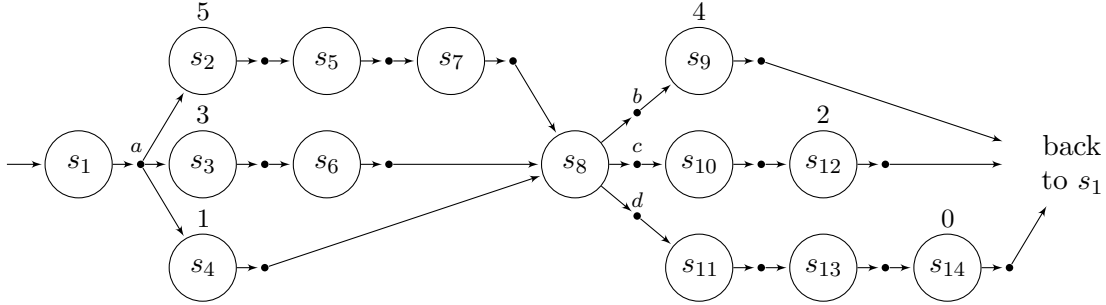
\begin{figure}[tbh]
  \begin{center}
  \begin{tikzpicture}[->,>=stealth',shorten >=1pt,auto,node
    distance=2.5cm,bend angle=45, scale=0.55, font=\normalsize]
    \tikzstyle{p1}=[draw,circle,text centered,minimum size=7mm,text width=5mm]
    \tikzstyle{p2}=[draw,rectangle,text centered,minimum size=7mm,text width=4mm]
    \tikzstyle{act}=[fill,circle,inner sep=1pt,minimum size=1.5pt, node distance=1cm]    \tikzstyle{empty}=[text centered, text width=15mm]
    \node[p1] (1) at (0,0) {$s_{1}$};
    \node[p1] (2) at (3,2.5) {$s_{2}$};
    \node[act] (2a) at (4.5,2.5) {};
    \node[p1] (2bis) at (6,2.5) {$s_{5}$};
    \node[act] (2bisa) at (7.5,2.5) {};
    \node[p1] (2tris) at (9,2.5) {$s_{7}$};
    \node[act] (2trisa) at (10.5,2.5) {};
    \node[p1] (3) at (3,0) {$s_{3}$};
    \node[act] (3a) at (4.5,0) {};
    \node[p1] (3bis) at (6,0) {$s_{6}$};
    \node[act] (3bisa) at (7.5,0) {};
    \node[p1] (4) at (3,-2.5) {$s_{4}$};
    \node[act] (4a) at (4.5,-2.5) {};
    \node[p1] (8) at (12,0) {$s_{8}$};
    \node[act] (8a) at (13.5,1.25) {};
    \node[act] (8b) at (13.5,0) {};
    \node[act] (8c) at (13.5,-1.25) {};
    \node[p1] (9) at (15,2.5) {$s_{9}$};
    \node[act] (9a) at (16.5,2.5) {};
    \node[empty] at ($(9)+(0,1.2)$) {\small 4};
    \node[p1] (10) at (15,0) {$s_{10}$};
    \node[act] (10a) at (16.5,0) {};
    \node[p1] (11) at (15,-2.5) {$s_{11}$};
    \node[act] (11a) at (16.5,-2.5) {};
    \node[p1] (10bis) at (18,0) {$s_{12}$};
    \node[act] (10bisa) at (19.5,0) {};
    \node[empty] at ($(10bis)+(0,1.2)$) {\small 2};
    \node[p1] (11bis) at (18,-2.5) {$s_{13}$};
    \node[act] (11bisa) at (19.5,-2.5) {};
    \node[p1] (11tris) at (21,-2.5) {$s_{14}$};
    \node[act] (11trisa) at (22.5,-2.5) {};
    \node[empty] at ($(11tris)+(0,1.2)$) {\small 0};
    \node[empty] (end) at (24,0) {back to $s_1$};
    \node[empty] at ($(2)+(0,1.2)$) {\small 5};
    \node[empty] at ($(3)+(0,1.2)$) {\small 3};
    \node[empty] at ($(4)+(0,1.2)$) {\small 1};
    \node[act] (1a) at (1.5,0) {};
    \node[empty] at ($(1a)+(-0.1,0.4)$) {\scriptsize $a$};
    \node[empty] at ($(8a)+(-0,0.4)$) {\scriptsize $b$};
    \node[empty] at ($(8b)+(-0,0.4)$) {\scriptsize $c$};
    \node[empty] at ($(8c)+(-0,0.4)$) {\scriptsize $d$};
    \coordinate[shift={(-5mm,0mm)}] (init) at (1.west);
    \path[-latex']
    (init) edge (1)
	(1) edge (1a)
	(2) edge (2a)
	(2a) edge (2bis)
	(3) edge (3a)
	(3a) edge (3bis)
	(4) edge (4a)
	(4a) edge (8)
	(8) edge (8a)
	(8) edge (8b)
	(8) edge (8c)
	(9) edge (9a)
	(9a) edge (end)
	(10) edge (10a)
	(10a) edge (10bis)
	(11) edge (11a)
	(11a) edge (11bis)
	(2bis) edge (2bisa)
	(2bisa) edge (2tris)
	(2tris) edge (2trisa)
	(2trisa) edge (8)
	(3bis) edge (3bisa)
	(3bisa) edge (8)
	(10bis) edge (10bisa)
	(10bisa) edge (end)
	(11bis) edge (11bisa)
	(11bisa) edge (11tris)
	(11tris) edge (11trisa)
	(11trisa) edge (end)
    (1a) edge (2)
    (1a) edge (3)
    (1a) edge (4)
    (8a) edge (9)
    (8b) edge (10)
    (8c) edge (11)
    ;
      \end{tikzpicture}
    \caption{This MDP with $d = 6$ priorities admits an almost-surely (even surely) winning strategy for $\DFW_\paritySub(\lambda = \frac{d}{2} + 2)$: the controller must answer to the path with priority $5$ (resp.~$3$, $1$) by choosing $b$ (resp.~$c$, $d$). This requires $\mathcal{O}(\frac{d}{2})$ memory states.}%
    \label{fig:memory}
  \end{center}
\end{figure}

\begin{exa}%
\label{ex:directmem}
Consider the direct fixed window parity objective in the MDP depicted in Fig.~\ref{fig:memory}, inspired by~\cite{DBLP:journals/corr/BruyereHR16}. For the sake of readability, we did not write all details in the figure: all actions have uniform probability distributions over their successors, and all unlabeled states have priority $d = 6$. This example can easily be generalized to any even $d \geq 0$. Observe that the MDP has size $\vert S\vert = 2 + \frac{d}{2} \cdot (\frac{d}{2} + 1)$, hence polynomial in $d$.

Fix the objective $\DFW_\paritySub(\lambda = \frac{d}{2} + 2)$. We claim that the controller may achieve it almost-surely (even surely) with a strategy using memory of size $\frac{d}{2}$. Indeed, each time action $a$ results in taking the path of priority $(d-1)$ (resp.~$(d-3), \ldots{}, 1$), the only possibility for the controller is to choose the path with priority $(d-2)$ (resp.~$(d-4), \ldots{}, 0$), otherwise a window stays open for $\lambda$ steps. Doing this ensures that the objective is satisfied surely.

Now, if the controller uses less than $\frac{d}{2}$ memory, he has to answer to two different odd priorities with the same choice of action in $s_8$, which results in a window staying open for $\lambda$ steps with strictly positive probability. Hence, there is no almost-sure winning strategy with such limited memory: polynomial-memory strategies are needed in this example.\hfill$\triangleleft$
\end{exa}

\section{The case of end-components}%
\label{sec:ec}
We have already solved the case of direct fixed window objectives: it remains to consider prefix-independent fixed and bounded variants. As seen in Sect.~\ref{sec:prelim}, the analysis of MDPs with prefix-independent objectives crucially relies on \textit{end-components} (ECs). Indeed, they are almost-surely reached in the long run.

In this section, we study what happens in ECs: how to play optimally and what can be achieved. In Sect.~\ref{sec:general}, we will use this knowledge as the cornerstone of our algorithm for general MDPs. The main result here is a \textit{strong link between ECs and two-player games}: intuitively, either the probability to win a window objective in an EC is zero, or it is one and there exists a sub-EC where the controller can actually win surely, i.e., as in a two-player game played on this sub-EC\@.

\paragraph{Safe ECs} We start by defining the notion of $\lambda$-safety, that will characterize such sub-ECs.

\begin{defi}[$\lambda$-safety]
  Let $\mdp$ be an MDP, $\Omega \in \{\mpSub, \paritySub\}$,
  $\lambda > 0$, and
  $\calC = (S_\calC,A_\calC,\delta_\calC) \in \ecs(\mdp)$, we say that $\calC$
  is $\lambda$-safe$_\Omega$ if there exists a strategy
 $\sigma \in \Sigma$ in $\calC$ such that, from all $s \in S_\calC$, $\sureStrat{\calC, s}{\sigma}{\DFW_\Omega(\lambda)}$.
\end{defi}

Classifying an EC as $\lambda$-safe$_{\Omega}$ or not boils down to
interpreting it as a \textit{two-player game} (the duality between MDPs and games is further explored in~\cite{DBLP:journals/iandc/BruyereFRR17,DBLP:conf/icalp/BerthonRR17}). In this setting, the
uncertainty becomes adversarial: when there is some uncertainty about
the outcome of an action, we do not consider probabilities but we let
the opponent decide the outcome of the action.  On entering a state
$s$ of the MDP after some history, the controller chooses an action
$a$ following its strategy and the
opponent then chooses a successor state $s'$ such that
$s' \in \supp(\delta(s, a))$ without taking into account the exact
values of probabilities.  In such a view, the opponent tries to
prevent the controller from achieving its objective.  A
\textit{winning strategy} for the controller in the game interpretation is a strategy that
ensures the objective regardless of its opponent's strategy. An EC is thus said to be $\lambda$-safe$_{\Omega}$ if and only if its two-player interpretation admits a winning strategy for $\DFW_\Omega(\lambda)$.

\begin{rem}
Throughout this paper, all the strategies we consider are uniform in the game-theoretic sense. That is, if an objective is achievable from a set of states, we do not need to have a different strategy for each starting state, and we may instead use the very same strategy from all starting states. This uniformity is not needed in our approach but we have it for free for both the classical MDP strategies (reachability, B\"uchi, etc)~\cite{baier2008principles} and the two-player window games~\cite{DBLP:journals/iandc/Chatterjee0RR15,DBLP:journals/corr/BruyereHR16}. We use it for the sake of readability, as it permits to have one strategy per EC instead of one per state of the EC for example. Hence all our statements are written with that in mind.
\end{rem}

\begin{prop}%
\label{prop:stratSafe}
Let $\mdp$ be an MDP, $\Omega \in \{\mpSub, \paritySub\}$, $\lambda > 0$, and $\calC = (S_\calC,A_\calC,\delta_\calC) \in \ecs(\mdp)$ be $\lambda$-safe$_{\Omega}$. Then, there exists a pure polynomial-memory strategy $\sigma^{\Omega, \lambda, \calC}_{\text{safe}}$ in $\calC$ such that $\sureStrat{\calC, s}{\sigma^{\Omega, \lambda, \calC}_{\text{safe}}}{\DFW_\Omega(\lambda)}$ for all $s \in S_\calC$.
\end{prop}
\begin{proof}
Straightforward by definition of $\lambda$-safety and pure polynomial-memory strategies being sufficient in two-player zero-sum direct fixed window games, both for mean-payoff~\cite{DBLP:journals/iandc/Chatterjee0RR15} and parity~\cite{DBLP:journals/corr/BruyereHR16} variants.
\end{proof}

\paragraph{Good ECs} As sketched above, the existence of sub-ECs that are $\lambda$-safe is crucial in order to satisfy any window objective in an EC\@. We thus introduce the notion of \textit{good} ECs.

\begin{defi}%
\label{def:good}
  Let $\mdp$ be an MDP, $\Omega \in \{\mpSub, \paritySub\}$, and
  $\calC \in \ecs(\mdp)$, we say that
\begin{itemize}
\item $\calC$ is $\lambda$-good$_\Omega$, for $\lambda > 0$, if it contains a sub-EC $\calC'$ which is $\lambda$-safe$_\Omega$.
\item $\calC$ is $\BW$-good$_\Omega$ if it contains a sub-EC $\calC'$ which is $\lambda$-safe$_\Omega$ for some $\lambda > 0$.
\end{itemize}
\end{defi}

\noindent
By definition, any $\BW$-good$_\Omega$ EC is also $\lambda$-good$_\Omega$ for an appropriate $\lambda > 0$. Yet, we use a different terminology as in the $\BW$ case, we do not fix $\lambda$ a priori: this will be important complexity-wise.

We now establish that good$_\Omega$ ECs are exactly the ones where window objectives can be satisfied with non-zero probability, and actually, with probability one.

\begin{lem}[Zero-one law]\label{grolem} Let $\mdp$ be an MDP\@,
  $\Omega \in \{\mpSub, \paritySub\}$ and
  $\calC = (S_\calC, A_\calC, \delta_\calC) \in \ecs(\mdp)$.  The
  following assertions hold.
  \begin{enumerate}[label={(\alph*)}]
  \item For all $\lambda > 0$, either\label{fixed-case-grolem}
    \begin{enumerate}[label={(\roman*)},left=-7mm]
    \item  $\calC$ is $\lambda$-good$_\Omega$ and
      there exists a strategy $\sigma$ in $\calC$ such that
      $\almostStrat{\calC, s}{\sigma}{\FW_\Omega(\lambda)}$ for all $s \in S_\calC$,\label{grolem1}
    \item or for all $s \in S_\calC$, for all strategy $\sigma$ in $\calC$,
      $\mathbb{P}_{\calC, s}^\sigma[\FW_\Omega(\lambda)] =
      0$.\label{grolem2}
    \end{enumerate}
  \item  Either\label{bounded-case-grolem}
    \begin{enumerate}[label={(\roman*)},left=-7mm]
    \item $\calC$ is
      $\BW$-good$_\Omega$ and
      there exists a strategy $\sigma$ in $\calC$ such that
      $\almostStrat{\calC, s}{\sigma}{\BW_\Omega}$ for all $s \in S_\calC$,\label{grolem3}
    \item or for all $s \in S_\calC$, for all strategy
      $\sigma$ in $\calC$, $\mathbb{P}_{\calC, s}^\sigma[\BW_\Omega] =
      0$.\label{grolem4}
    \end{enumerate}
  \end{enumerate}
\end{lem}

\begin{proof}
  We begin with the fixed variant~\ref{fixed-case-grolem}.  Case~\ref{grolem1}. Fix
  $\lambda > 0$ and assume there exists a sub-EC $\calC'$ with state
  space $S_{\calC'}$ that is $\lambda$-safe$_\Omega$. By Prop.~\ref{prop:stratSafe}, there
  exists a strategy $\sigma^{\Omega, \lambda, \calC'}_{\text{safe}}$ in $\calC'$ such that
  $\sureStrat{\calC', s}{\sigma^{\Omega, \lambda, \calC'}_{\text{safe}}}{\DFW_\Omega(\lambda)}$ for all $s \in S_{\calC'}$. Now, since $\calC$ is an EC, there exists a (pure memoryless) strategy $\sigma_{\text{reach}}$ in $\calC$ that ensures eventually reaching $\calC'$ almost-surely from any state $s \in S_\calC$~\cite{baier2008principles}. Hence, the desired strategy $\sigma$ can be defined as follows: play according to $\sigma_{\text{reach}}$ until $\calC'$ is reached, then switch to $\sigma^{\Omega, \lambda, \calC'}_{\text{safe}}$ forever. It is straightforward to check that $\sigma$ almost-surely satisfies $\FW_\Omega(\lambda)$ from anywhere in $\calC$, thanks to prefix-independence. Note that it does not ensure it surely in general, as $\sigma_{\text{reach}}$ does not guarantee to reach $\calC'$ surely either.

  Case~\ref{grolem2}. Now assume that such a $\lambda$-safe$_{\Omega}$ sub-EC does not exist.  Recall
  that finite-memory strategies suffice for
  $\FW_\Omega(\lambda)$ objectives by Thm.~\ref{finite-memory-cor}, hence we can restrict our study to such strategies without loss of generality. Fix any finite-memory strategy
  $\sigma$ in $\calC$ and state $s \in S_\calC$. The induced MC $\mdp^\sigma_s$ is finite, hence runs almost-surely end up in a BSCC~\cite{baier2008principles}. Let $\mathcal{B}$ be any BSCC of $\mdp^\sigma_s$ reached with positive probability. Since there exists no $\lambda$-safe$_{\Omega}$ sub-EC in $\calC$, there must exist a run $\widehat{\rho}$ in $\mathcal{B}$ such that $\widehat{\rho} \not\in \DFW_\Omega(\lambda)$; otherwise, $\sigma$ would be witness that the EC obtained by projecting $\mathcal{B}$ over $S_\calC$ is $\lambda$-safe$_{\Omega}$. From $\widehat{\rho}$, we extract a history $\widehat{h}$ ending with a window open for $\lambda$ steps (it exists otherwise $\widehat{\rho}$ would be in $\DFW_\Omega(\lambda)$). This history is finite: it has a probability lower-bounded by some $\varepsilon > 0$ to happen whenever its starting state is visited. Now, since all states in $\mathcal{B}$ are almost-surely visited infinitely often, we conclude that this history also happens infinitely often with probability one. Therefore, the probability to win the prefix-independent objective $\FW_\Omega(\lambda)$ when reaching $\mathcal{B}$ is zero. Since this holds for any BSCC induced by $\sigma$, we obtain the claim.

   Let us continue with the bounded case~\ref{bounded-case-grolem}.  Case~\ref{grolem3} is trivial thanks to~\ref{fixed-case-grolem}\ref{grolem1} and $\FW_\Omega(\lambda) \subseteq \BW_{\Omega}$. Now consider case~\ref{grolem4}. By~\ref{fixed-case-grolem}\ref{grolem2}, we have that $\mathbb{P}_{\calC, s}^\sigma[\FW_\Omega(\lambda)] = 0$ for all $s \in S_\calC$, $\lambda > 0$ and $\sigma$ in $\calC$. Observe that by definition, we have
  \begin{equation*}
    \BW_{\Omega} = \bigcup_{\lambda > 0} \FW_{\Omega}(\lambda).
  \end{equation*}
Fix any strategy $\sigma$ in $\calC$ and $s \in S_\calC$. By countable additivity of probability measures, we have that
  \begin{equation*}
    \mathbb{P}_{\calC, s}^\sigma[\BW_\Omega] \; \leq \; \sum_{\lambda > 0}
    \mathbb{P}^\sigma_{\calC, s} [\FW_\Omega(\lambda)].
  \end{equation*}
Since the right term is equal to zero, we obtain the claim.
\end{proof}

\begin{rem}%
\label{rmk:uniformECs}
An interesting consequence of Lem.~\ref{grolem} is the existence of uniform bounds on $\lambda$ in ECs, in contrast to the general MDP case, as seen in Sect.~\ref{subsec:illustration}. This is indeed natural, as we established that winning with positive probability within an EC coincides with winning surely in a sub-EC;\@ sub-EC that can be seen as a two-player zero-sum game where uniform bounds are granted by~\cite{DBLP:journals/iandc/Chatterjee0RR15,DBLP:journals/corr/BruyereHR16}.
\end{rem}

By Lem.~\ref{grolem}, we know that interesting strategies exist in good$_ \Omega$ ECs. Let us describe them.

\begin{prop}%
\label{prop:stratGood}
Let $\mdp$ be an MDP, $\Omega \in \{\mpSub, \paritySub\}$, and $\calC = (S_\calC,A_\calC,\delta_\calC) \in \ecs(\mdp)$.
\begin{itemize}
\item If $\calC$ is $\lambda$-good$_\Omega$, for some $\lambda > 0$, there exists a pure polynomial-memory strategy $\sigma^{\Omega, \lambda, \calC}_{\text{good}}$ such that $\almostStrat{\calC, s}{\sigma^{\Omega, \lambda, \calC}_{\text{good}}}{\FW_\Omega(\lambda)}$ for all $s \in S_\calC$.
\item If $\calC$ is $\BW$-good$_\Omega$, there exists a pure memoryless strategy $\sigma^{\Omega, \BW, \calC}_{\text{good}}$ such that $\almostStrat{\calC, s}{\sigma^{\Omega, \BW, \calC}_{\text{good}}}{\BW_\Omega}$ for all $s \in S_\calC$.
\end{itemize}
\end{prop}

\noindent
Intuitively, such strategies first mimic a pure memoryless strategy reaching a safe$_{\Omega}$ sub-EC almost-surely, then switch to a strategy surely winning in this sub-EC, which is lifted from the game interpretation.
\begin{proof}
Consider the $\lambda$-good$_\Omega$ case. Strategy $\sigma^{\Omega, \lambda, \calC}_{\text{good}}$  is the one described in the proof of Lem.~\ref{grolem}\ref{fixed-case-grolem}\ref{grolem1}: it first plays as the pure memoryless strategy $\sigma_\text{reach}$ then switches to strategy $\sigma^{\Omega, \lambda, \calC'}_{\text{safe}}$ from Prop.~\ref{prop:stratSafe} when the $\lambda$-safe$_\Omega$ sub-EC $\calC'$ is reached. Hence it is pure and polynomial memory still suffices. The almost-sure satisfaction of the objective was proved in Lem.~\ref{grolem}.

Now, consider the $\BW$-good$_\Omega$ case. Note that using our current knowledge, it is easy to build a pure  polynomial-memory strategy, but we want more: a pure \textit{memoryless} one. For that, we use the following reasoning. Strategy $\sigma^{\Omega, \BW, \calC}_{\text{good}}$  will again first play as a pure memoryless strategy trying to reach a $\lambda$-safe$_\Omega$ sub-EC for some $\lambda > 0$ (we know one exists), then switch to a surely winning strategy inside this sub-EC\@. Still, we do not really need to win for $\DFW_{\Omega}(\lambda)$, we only need to win for the bounded variant $\DBW_{\Omega} = \bigcup_{\lambda > 0} \DFW_{\Omega}(\lambda)$. Results in two-player zero-sum games (which we are actually solving here when considering surely winning strategies) guarantee that in this case, pure memoryless strategies suffice for the controller, both for mean-payoff~\cite{DBLP:journals/iandc/Chatterjee0RR15} and parity variants~\cite{DBLP:journals/corr/BruyereHR16}. Since the switch is state-based (depending on whether we are in the sub-EC or not), we can wrap both strategies in $\sigma^{\Omega, \BW, \calC}_{\text{good}}$ and retain a pure memoryless strategy. Again, the almost-sure satisfaction of the objective follows from Lem.~\ref{grolem}.
\end{proof}

\paragraph{Classification} We may already sketch a general solution to the threshold probability problem based on Lem.~\ref{grolem} and the well-known fact that ECs are almost-surely reached under any strategy: an optimal strategy must maximize the probability to reach good$_\Omega$ ECs. It is therefore crucial to be able to identify such ECs efficiently. However, an MDP may in general contain an exponential number of ECs. Fortunately, the next lemma establishes that we do not have to test them all.

\begin{lem}
Let $\mdp$ be an MDP and $\calC \in \ecs(\mdp)$. If $\calC$ is $\lambda$-good$_\Omega$ (resp.~$\BW$-good$_\Omega$), then it is also the case of any super-EC $\calC' \in \ecs(\mdp)$ containing $\calC$.
\end{lem}
\begin{proof}
Trivial by Def.~\ref{def:good}.
\end{proof}
\begin{cor}%
\label{cor:goodMECs}
Let $\mdp$ be an MDP and $\calC \in \mecs(\mdp)$ be a maximal EC\@. If $\calC$ is not $\lambda$-good$_\Omega$ (resp.~$\BW$-good$_\Omega$), then neither is any of its sub-EC $\calC' \in \ecs(\mdp)$.
\end{cor}

The interest of Cor.~\ref{cor:goodMECs} is that the number of MECs is bounded by $\vert S\vert$ for any MDP $\mdp = (S,A,\delta)$ because they are all disjoints. Furthermore, the MEC decomposition can be done efficiently (e.g., quadratic time~\cite{DBLP:journals/jacm/ChatterjeeH14}). So, we know classifying MECs is sufficient and MECs can easily be identified in an MDP:\@ it remains to discuss how to classify a MEC as good$_\Omega$ or not.

Let $\mdp = (S,A,\delta)$. Recall that a MEC $\calC = (S_\calC,A_\calC,\delta_\calC) \in \mecs(\mdp)$ is $\lambda$-good$_\Omega$ (resp.~$\BW$-good$_\Omega$) if and only if it contains a $\lambda$-safe$_\Omega$ sub-EC\@. By definition of $\lambda$-safety, this is equivalent to having a non-empty \textit{winning set} for the controller in the two-player zero-sum game over $\calC$ --- naturally defined as explained above. This winning set contains all states in $S_\calC$ from which the controller has a \textit{surely winning} strategy. This winning set, if non-empty, necessarily contains at least one sub-EC of $\calC$, as otherwise the opponent could force the controller to leave it and win the game (by prefix-independence). Thus, testing if a MEC is good$_\Omega$ boils down to solving its two-player game interpretation.

\begin{thm}[MEC classification]%
\label{thm:classification}
Let $\mdp$ be an MDP and $\calC \in \mecs(\mdp)$. The following assertions hold.
  \begin{enumerate}[label={(\alph*)}]
  \item Deciding if $\calC$ is $\lambda$-good$_\Omega$, for $\lambda > 0$, is in $\mathsf{P}$ for $\Omega \in \{\mpSub, \paritySub\}$. Furthermore, a corresponding pure polynomial-memory strategy $\sigma^{\Omega, \lambda, \calC}_{\text{good}}$ can be constructed in polynomial time.
  \item Deciding if $\calC$ is $\BW$-good$_\mpSub$ is in $\NPTIME \cap \mathsf{coNP}$ and a corresponding pure memoryless strategy $\sigma^{\mpSub, \BW, \calC}_{\text{good}}$ can be constructed in pseudo-polynomial time.
  \item Deciding if $\calC$ is $\BW$-good$_\paritySub$ is in $\PTIME$ and a corresponding pure memoryless strategy $\sigma^{\paritySub, \BW, \calC}_{\text{good}}$ can be constructed in polynomial time.
  \end{enumerate}
\end{thm}
\begin{proof}
All complexities are expressed w.r.t.\ the representation of $\calC$, as the larger MDP $\mdp$ is never used in the process. Complexity results follow from game solving algorithms presented in~\cite{DBLP:journals/iandc/Chatterjee0RR15} for mean-payoff\footnote{The published version of~\cite{DBLP:journals/iandc/Chatterjee0RR15} contains a slight bug in sub-procedure $\mathsf{GoodWin}$ that was corrected in subsequent articles~\cite{DBLP:conf/concur/BruyereHR16,DBLP:journals/corr/BruyereHR16}. All results of~\cite{DBLP:journals/iandc/Chatterjee0RR15} still hold modulo this correction.} and~\cite{DBLP:journals/corr/BruyereHR16} for parity. Note that for mean-payoff, the bound on $\lambda$ for which the fixed and bounded variants coincide is pseudo-polynomial, whereas for parity this bound is polynomial, hence the different results: the algorithm for bounded window mean-payoff games actually bypasses this bound to obtain $\NPTIME \cap \mathsf{coNP}$ membership instead of simply $\mathsf{EXPTIME}$.

For constructing the strategy, we only need to plug the almost-sure-reachability strategy to the surely-winning one, as presented in Prop.~\ref{prop:stratGood}. Such a reachability strategy can easily be computed in polynomial time~\cite{baier2008principles}, hence the complexity is dominated by the cost of computing the surely-winning one, as presented in~\cite{DBLP:journals/iandc/Chatterjee0RR15,DBLP:journals/corr/BruyereHR16}.
\end{proof}

\section{General MDPs}%
\label{sec:general}

\paragraph{Algorithms} We now have all the ingredients to establish an algorithm for the threshold probability problem in the general case. Intuitively, given an MDP $\mdp$, an initial state $s$ and a window objective $\FW_{\Omega}(\lambda)$ for $\lambda > 0$ (resp.~$\BW_{\Omega}$), we first compute the MEC decomposition of $\mdp$, then classify each MEC as $\lambda$-good$_\Omega$ (resp.~$\BW$-good$_\Omega$) or not, and finally compute an optimal strategy from $s$ to reach the union of good$_\Omega$ MECs: the probability of reaching such MECs is then exactly the maximum probability to satisfy the window objective.

The fixed and bounded versions are presented in Fig.~\ref{fig:algo}. Let $\mdp = (S, A, \delta)$ be the MDP\@. The MEC decomposition (Line~\ref{line:decomp}) takes quadratic time~\cite{DBLP:journals/jacm/ChatterjeeH14} and yields at most $\vert S\vert$ MECs. The classification depends on the variant considered, as established in Thm.~\ref{thm:classification}: it is in $\PTIME$ for fixed variants and bounded window parity, and in $\NPTIME \cap \mathsf{coNP}$ for bounded window mean-payoff (with a pseudo-polynomial-time procedure). Finally, sub-procedure $\mathsf{MaxReachability}(s, T)$ computes the maximum probability to reach the set $T$ from $s$. It is well-known that this can be done in polynomial time and that pure memoryless optimal strategies exist~\cite{baier2008principles}. Therefore the overall complexity of the algorithm is dominated by the classification step.

\renewcommand{\algorithmicrequire}{\textbf{Input:}}
\renewcommand{\algorithmicensure}{\textbf{Output:}}
\begin{figure}[tbh]
\begin{center}
\begin{minipage}[t]{.48\linewidth}
\begin{algorithm}[H]
  \caption{\label{alg:FW}$\mathsf{FixedWindow}(\mdp, s, \Omega, \lambda)$}
  \algsetup{linenodelimiter=}
  \begin{algorithmic}[1]
    \REQUIRE{MDP $\mdp$, state $s$, $\Omega
      \in \{\mpSub, \paritySub\}$, $\lambda > 0$}
    \ENSURE{Maximum probability of
      $\FW_\Omega(\lambda)$ from $s$}
    \STATE{$T\gets \emptyset$}
    \FORALL{\label{line:decomp} $\calC = (S_\calC, A_\calC, \delta_\calC) \in \mecs(\mdp)$}
    \IF{\label{line:classify} $\calC$ is $\lambda$-good$_{\Omega}$}
    \STATE{$T \gets T \uplus S_\calC$}
    \ENDIF%
    \ENDFOR%
    \STATE{\label{line:reach} $\nu = \mathsf{MaxReachability}(s, T)$}
    \RETURN{$\nu$}
  \end{algorithmic}
\end{algorithm}
\end{minipage}
\hfill
\begin{minipage}[t]{.48\linewidth}
\begin{algorithm}[H]
  \caption{\label{alg:BW}$\mathsf{BoundedWindow}(\mdp, s, \Omega)$}
  \algsetup{linenodelimiter=}
  \begin{algorithmic}[1]
    \REQUIRE{MDP $\mdp$, state $s$, $\Omega
      \in \{\mpSub, \paritySub\}$}
    \ENSURE{Maximum probability of
      $\BW_\Omega$ from $s$}
    \STATE{$T\gets \emptyset$}
    \FORALL{$\calC = (S_\calC, A_\calC, \delta_\calC) \in \mecs(\mdp)$}
    \IF{$\calC$ is $\BW$-good$_{\Omega}$}
    \STATE{$T \gets T \uplus S_\calC$}
    \ENDIF%
    \ENDFOR%
    \STATE{$\nu = \mathsf{MaxReachability}(s, T)$}
    \RETURN{$\nu$}
  \end{algorithmic}
\end{algorithm}
\end{minipage}
\caption{Algorithms computing the maximum probability of fixed and bounded window objectives in general MDPs.}%
\label{fig:algo}
\end{center}
\end{figure}

\paragraph{Correctness} We first prove that these algorithms are sound and complete.

\begin{lem}\label{lm:correctness}
Alg.~\ref{alg:FW} and Alg.~\ref{alg:BW} are correct: given an MDP $\mdp = (S, A, \delta)$, an initial state $s \in S$, $\Omega \in \{\mpSub, \paritySub\}$, $\lambda > 0$, we have that
\begin{align*}
\mathsf{FixedWindow}(\mdp, s, \Omega, \lambda) &= \max_{\sigma \in \Sigma}  \pr_{\mdp,s}^\sigma[\FW_{\Omega}(\lambda)],\\
\mathsf{BoundedWindow}(\mdp, s, \Omega) &= \max_{\sigma \in \Sigma}  \pr_{\mdp,s}^\sigma[\BW_{\Omega}].
\end{align*}
\end{lem}

\begin{proof}
The proof is straightforward based on our previous results. First, recall that objectives $\FW_{\Omega}(\lambda)$ and $\BW_{\Omega}$ are prefix-independent; that in MDPs, the limit-behavior under any strategy almost-surely coincides with an EC (see Sect.~\ref{sec:prelim}); and that MECs are pair-wise disjoint. We thus have that
\begin{align*}
\sup_{\sigma \in \Sigma}  \pr_{\mdp,s}^\sigma[\FW_{\Omega}(\lambda)] &= \sup_{\sigma \in \Sigma}  \sum_{\calC \in \mecs(\mdp)} \pr_{\mdp,s}^\sigma[\diamondsuit\square \calC] \cdot \pr_{\calC}^\sigma[\FW_{\Omega}(\lambda)],\\
\sup_{\sigma \in \Sigma} \pr_{\mdp,s}^\sigma[\BW_{\Omega}] &= \sup_{\sigma \in \Sigma} \sum_{\calC \in \mecs(\mdp)} \pr_{\mdp,s}^\sigma[\diamondsuit\square \calC] \cdot \pr_{\calC}^\sigma[\BW_{\Omega}],
\end{align*}
where $\diamondsuit\square \mathcal{C}$ uses the standard LTL notation as a shorthand for \[\{\rho = s_0 a_0 s_1 \ldots \in \Runs{\mdp} \mid \exists\, i \geq 0,\, \forall\, j > i,\, s_j \in S_\calC\}\] with $S_\calC$ the set of states of $\calC$; and where we write $\pr_{\calC}$ without distinction on the entry point because the supremum probability to win for a prefix-independent objective is identical in all states of a MEC (as the controller may force reaching any state he wants almost-surely).

By Lem.~\ref{grolem}, we know that \textit{supremum} winning probabilities in ECs are zero-one laws: they are equal to one in good$_\Omega$ ECs and to zero in all other ECs. Furthermore, when almost-sure satisfaction is achievable, the corresponding strategy stays in the EC\@. Let us denote by $\mathsf{LGEC}_\Omega(\mdp, \lambda)$ and $\mathsf{BWGEC}_\Omega(\mdp)$ the sets of $\lambda$-good$_\Omega$ and $\BW$-good$_\Omega$ MECs of $\mdp$ respectively, and by $L$ and $B$ the disjoint union of the corresponding state spaces. That is,
\begin{align*}
L = \biguplus_{\calC = (S_\calC, A_\calC, \delta_\calC) \in \mathsf{LGEC}_\Omega(\mdp, \lambda)} S_\calC, \qquad \qquad \qquad\qquad
B = \biguplus_{\calC = (S_\calC, A_\calC, \delta_\calC) \in \mathsf{BWGEC}_\Omega(\mdp)} S_\calC.
\end{align*}
We have that
\begin{align*}
\sup_{\sigma \in \Sigma} \pr_{\mdp,s}^\sigma[\FW_{\Omega}(\lambda)] &= \sup_{\sigma \in \Sigma} \sum_{\calC \in \mathsf{LGEC}_\Omega(\mdp, \lambda)} \pr_{\mdp,s}^\sigma[\diamondsuit \square \calC] = \sup_{\sigma \in \Sigma} \pr_{\mdp,s}^\sigma\left[\diamondsuit L\right],\\
\sup_{\sigma \in \Sigma} \pr_{\mdp,s}^\sigma[\BW_{\Omega}] &= \sup_{\sigma \in \Sigma} \sum_{\calC \in \mathsf{BWGEC}_\Omega(\mdp)} \pr_{\mdp,s}^\sigma[\diamondsuit \square\calC]  = \sup_{\sigma \in \Sigma} \pr_{\mdp,s}^\sigma\left[\diamondsuit B\right].
\end{align*}

Now, observe that the term on the right-hand side is exactly the probability computed by our algorithm through the call to sub-procedure $\mathsf{MaxReachability}$ on Line~\ref{line:reach}: our algorithm is thus correct, as it computes the supremum probability to achieve the corresponding window objective.

It remains to establish that this probability is actually a \textit{maximum} probability, as claimed. Indeed, there exists a strategy to achieve the probability $\nu$ returned by our algorithm. From $s$, we know that pure memoryless optimal strategies exist to reach the computed set $T$ (which coincides with $L$ or $B$ depending on the considered variant), and inside $T$, there exist pure finite-memory optimal strategies for all good$_\Omega$ MECs, as shown in Thm.~\ref{thm:classification}. We can easily combine all these strategies in an optimal strategy for $\mdp$, which concludes our proof.
\end{proof}

\paragraph{Complexity} We wrap-up our results on prefix-independent window objectives in the following theorem.

\begin{thm}%
\label{thm:general}
The threshold probability problem is
\begin{enumerate}[label={(\alph*)}]
\item in $\PTIME$ for fixed window parity objectives and fixed window mean-payoff objectives, and pure polynomial-memory optimal strategies can be constructed in polynomial time;
\item in $\PTIME$ for bounded window parity objectives, and pure memoryless optimal strategies can be constructed in polynomial time;
\item in $\NPTIME \cap \mathsf{coNP}$ for bounded window mean-payoff objectives, and pure memoryless optimal strategies can be constructed in pseudo-polynomial time.
\end{enumerate}
\end{thm}
\begin{proof}
The complexity results follow from the analysis we made before:
\begin{itemize}
\item the MEC decomposition (Line~\ref{line:decomp}) takes quadratic time~\cite{DBLP:journals/jacm/ChatterjeeH14} and yields at most $\vert S\vert$ MECs;
\item classifying a MEC (Line~\ref{line:classify}) is in $\PTIME$ for fixed variants and bounded window parity, and in $\NPTIME \cap \mathsf{coNP}$ for bounded window mean-payoff (Thm.~\ref{thm:classification});
\item $\mathsf{MaxReachability}(s, T)$ (Line~\ref{line:reach}) takes polynomial time~\cite{baier2008principles}.
\end{itemize}
Overall we have $\PTIME$-membership for all variants except bounded window mean-payoff, where we are in $\PTIME^{\NPTIME \cap \mathsf{coNP}}$, which is equal to $\NPTIME \cap \mathsf{coNP}$~\cite{DBLP:journals/tit/Brassard79}.

Regarding optimal strategies, they are by-products of Alg.~\ref{alg:FW} and Alg.~\ref{alg:BW}. Inside good$_\Omega$ MECs, we play the strategy granted by Thm.~\ref{thm:classification}, and outside, we simply play a pure memoryless optimal reachability strategy obtained through $\mathsf{MaxReachability}(s, T)$~\cite{baier2008principles}.
\end{proof}

\paragraph{Lower bounds} We complement the results of Thm.~\ref{thm:general} with matching lower bounds, showing that our approach is optimal complexity-wise.
\begin{thm}%
\label{thm:generalHardness}
The threshold probability problem is
\begin{enumerate}[label={(\alph*)}]
\item $\PTIME$-hard for fixed window parity objectives and fixed window mean-payoff objectives, and polynomial-memory strategies are in general necessary;
\item $\PTIME$-hard for bounded window parity objectives;
\item as hard as mean-payoff games for bounded window mean-payoff objectives.
\end{enumerate}
\end{thm}
\begin{proof}
Complexity-wise, the lower bounds follow from the results in two-player zero-sum window mean-payoff~\cite{DBLP:journals/iandc/Chatterjee0RR15} and window parity~\cite{DBLP:journals/corr/BruyereHR16} games, coupled with the equivalence between the threshold probability problem and solving a two-player game established in Sect.~\ref{sec:ec}. Note that formally, this equivalence only holds if the MDP is an EC, hence we have to make sure that the corresponding game problems retain their hardness when considered over arenas that correspond to ECs (i.e., when the underlying graph is strongly connected). Fortunately, careful inspection of the related results in~\cite{DBLP:journals/iandc/Chatterjee0RR15,DBLP:journals/corr/BruyereHR16} shows that this is the case. Note that mean-payoff games are widely considered as a canonical ``hard'' problem for the class $\NPTIME \cap \mathsf{coNP}$ (see, e.g.,~\cite{DBLP:journals/iandc/Chatterjee0RR15}).

Regarding strategies, the fixed case is the only one where memory is needed, as stated in Thm.~\ref{thm:general}. Again, the necessity of polynomial memory is witnessed through the equivalence with two-player games in ECs, proved in Sect.~\ref{sec:ec}, along with results on fixed window mean-payoff and fixed window parity games~\cite{DBLP:journals/iandc/Chatterjee0RR15,DBLP:journals/corr/BruyereHR16}.
\end{proof}

\section{Limitations and perspectives}%
\label{sec:concl}
Recall that we summarized our results in Table~\ref{table:overview}, and compared them to the state of the art in Sect.~\ref{sec:intro}. The goal of this section is to discuss the limitations of our work and some extensions within arm's reach.

\paragraph{Direct bounded window objectives} We left out a specific variant of window objectives that was considered in games~\cite{DBLP:journals/iandc/Chatterjee0RR15,DBLP:journals/corr/BruyereHR16}: the \textit{direct bounded} variant, defined as $\DBW_\Omega = \bigcup_{\lambda > 0} \DFW_{\Omega}(\lambda)$ (see Sect.~\ref{subsec:illustration}). This variant is maybe not the most natural as it is \textit{not} prefix-independent, yet allows to close the windows of a run in an arbitrarily large --- but bounded along the run --- number of steps.

This variant gives rise to complex behaviors in MDPs, notably due to its interaction with the almost-sure reachability of ECs. Let us illustrate it on an example.

\begin{exa}
Consider the MDP in Fig.~\ref{fig:DBW} and objective $\DBW_\mpSub$. A window opens in the first step due to action $a$. The only way to close it is to loop in $s_2$ using action $b$ up to the point where the running sum of weights becomes non-negative. Note also that when it does, all windows are closed and the controller may safely switch to $s_5$. Now, observe that taking action $b$ repeatedly induces a \textit{symmetric random walk}~\cite{grinstead_AMS1997}. Classical probability results ensure that a non-negative sum will be obtained almost-surely, but the number of times $b$ is played must remain unbounded (as for any bounded number, there exists a strictly positive probability to obtain only $-1$'s for example). Therefore, in this example, there exists an infinite-memory strategy $\sigma$ such that $\almostStrat{\mdp,s_1}{\sigma}{\DBW_\mpSub}$, but no finite-memory strategy can do as good.

Now, if we modify the probabilities in $b$ to $\delta(s_2, b) = \{s_3 \mapsto 0.6, s_4 \mapsto 0.4\}$, the random walk becomes \textit{asymmetric}, with a strictly positive chance to diverge toward $-\infty$. While the best possible strategy is still the one defined above, it only guarantees a probability strictly less than one to satisfy the objective.\hfill$\triangleleft$
\end{exa}

\begin{figure}[tbh]
\centering
  \begin{tikzpicture}[->,>=stealth',shorten >=1pt,auto,node
    distance=2.5cm,bend angle=45, scale=0.6, font=\normalsize]
    \tikzstyle{p1}=[draw,circle,text centered,minimum size=7mm,text width=4mm]
    \tikzstyle{p2}=[draw,rectangle,text centered,minimum size=7mm,text width=4mm]
    \tikzstyle{act}=[fill,circle,inner sep=1pt,minimum size=1.5pt, node distance=1cm]    \tikzstyle{empty}=[text centered, text width=15mm]
    \node[p1] (1) at (0,0) {$s_{1}$};
    \node[p1] (2) at (4,0) {$s_{2}$};
    \node[act] (1a) at (2,0) {};
    \node[act] (2a) at (4,3) {};
    \node[act] (2b) at (6,0) {};
    \node[act] (3a) at (2,1.5) {};
    \node[act] (4a) at (6,1.5) {};
    \node[act] (5a) at (10,0) {};
    \node[empty] at ($(3a)+(-0.7,0)$) {\scriptsize $c, -1$};
    \node[empty] at ($(4a)+(0.5,0)$) {\scriptsize $d, 1$};
    \node[p1] (3) at (2,3) {$s_{3}$};
    \node[p1] (4) at (6,3) {$s_{4}$};
    \node[empty] at ($(1a)+(-0,0.3)$) {\scriptsize $a, -1$};
    \node[empty] at ($(5a)+(-0,-0.3)$) {\scriptsize $f, 0$};
    \node[empty] at ($(2a)+(-0,0.3)$) {\scriptsize $b, 0$};
    \node[empty] at ($(2b)+(-0,0.3)$) {\scriptsize $e, 0$};
    \node[p1] (5) at (8,0) {$s_{5}$};
    \coordinate[shift={(-5mm,0mm)}] (init) at (1.west);
    \path[-latex']
    (init) edge (1)
	(1) edge (1a)
	(2) edge (2a)
	(2) edge (2b)
	(3) edge (3a)
	(5) edge (5a)
	(4) edge (4a)
	(1a) edge node[below,xshift=0mm]{\scriptsize $1$} (2)
	(2b) edge node[below,xshift=0mm]{\scriptsize $1$} (5)
	(3a) edge node[right,xshift=1mm]{\scriptsize $1$} (2)
	(4a) edge node[left,xshift=-1mm]{\scriptsize $1$} (2)
	(2a) edge node[below,xshift=0mm]{\scriptsize $0.5$} (3)
	(2a) edge node[below,xshift=0mm]{\scriptsize $0.5$} (4)
    ;
\draw [->] (5a) to[out=90,in=45] node[above ,xshift=0mm]{\scriptsize $1$} (5);
      \end{tikzpicture}
    \caption{There exists a strategy $\sigma$ ensuring $\almostStrat{\mdp,s_1}{\sigma}{\DBW_\mpSub}$ but it requires infinite memory as it needs to use $b$ up to the point where the running sum becomes non-negative, then switch to $e$.}%
    \label{fig:DBW}
\end{figure}
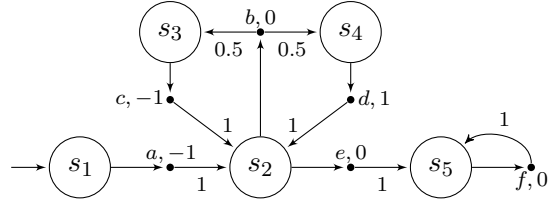

What do we observe? First, \textit{infinite-memory strategies are required}, which is a problem for practical applications. Second, even for \textit{qualitative} questions (is the probability zero or one?), the actual probabilities of the MDP must be considered, not only the existence of a transition. This is in stark contrast to most problems in MDPs~\cite{baier2008principles}: in that sense, the direct bounded window objective is not well-behaved. This is due to the connection with random walks we just established. Imagine that we replace the gadget related to action $b$ by a much more complex EC:\@ the corresponding random walk will be quite tedious to analyze. It is well-known that complex random walks are difficult to tackle for verification and synthesis. For example, even simple asymmetric random walks, like the one we sketched, are not \textit{decisive} MCs, a large and robust class of MCs where reachability questions can be answered~\cite{DBLP:journals/lmcs/AbdullaHM07}.

\paragraph{Markov chains} Our work focuses on the threshold probability problem for MDPs, and the corresponding strategy synthesis problem. Better complexities could possibly be obtained in the case of MCs, where there are no non-deterministic choices. To achieve this, a natural direction would be to focus on the classification of ECs (Sect.~\ref{sec:ec}), the bottleneck of our approach: for MCs, this classification would involve \textit{one-player window games} (for the opponent), whose precise complexity has yet to be explored and would certainly be lower than for two-player games.

Note however that complexity classes in the MC case are unlikely to be much lower: all parity variants are already in $\mathsf{P}$, and the high complexity of the direct fixed window mean-payoff case would remain: a construction similar to the $\mathsf{PSPACE}$-hardness proof (Thm.~\ref{thm:directHardness}) easily shows this problem to be $\mathsf{PP}$-hard, already for \textit{acyclic} MCs (again using~\cite{DBLP:conf/icalp/HaaseK15}). Let us recall that $\mathsf{PP}$-hard problems are widely believed to be outside $\mathsf{NP}$, as otherwise the polynomial hierarchy would collapse to $\mathsf{P}^\mathsf{NP}$ by Toda's theorem~\cite{DBLP:journals/siamcomp/Toda91}.

\paragraph{Expected value problem} Given an MDP $\mdp$ and an initial state $s$, we may be interested in synthesizing a strategy $\sigma$ that minimizes the \textit{expected window size} for a fixed window objective (say $\FW_\Omega(\lambda)$ for the sake of illustration), which we straightforwardly define as
\[
\mathbb{E}^\sigma_{\mathcal{M}, s, \Omega}(\lambda) = \sum_{\lambda > 0}^\infty \lambda \cdot  \mathbb{P}_{\mathcal{M}, s}^\sigma[\FW_\Omega(\lambda) \setminus \FW_\Omega(\lambda-1)],
\]
with $\FW_\Omega(0) = \emptyset$. This meets the natural desire to build strategies that strive to maintain the \textit{best time bounds possible in their local environment} (e.g., EC of $\mdp$). Note that this is totally different from the value function used in~\cite{DBLP:journals/corr/abs-1812-09298}.

For prefix-independent variants, \textit{we already have all the necessary machinery} to solve this problem. First, we refine the classification process to identify the best window size achievable in each MEC, if any. Indeed, if a MEC is $\lambda$-good, it necessarily is for some $\lambda$ between one and the upper bound derived from the game-theoretic interpretation (Rmk.~\ref{rmk:uniformECs}): we determine the smallest value of $\lambda$ for each MEC via a binary search coupled with the classification procedure. Second, using classical techniques (e.g.,~\cite{DBLP:journals/fmsd/RandourRS17}), we contract each MEC to a single-state EC, and give it a weight that represents the best window size we can ensure in it (hence this weight may be infinite if a MEC is not $\BW$-good). Finally, we construct a global strategy that favors reaching MECs with the lowest weights, for example by synthesizing a strategy minimizing the classical mean-payoff value. Note that if $\lambda$-good MECs cannot be reached almost-surely, the expected value will be infinite, as wanted. Observe that \textit{such an approach maintains tractability}, as we end up with a polynomial-time algorithm.

Direct variants would require more involved techniques, as the unfoldings developed in Sect.~\ref{sec:fixed} are strongly linked to the fixed window size $\lambda$, and cannot be that easily combined for different values of $\lambda$.

\paragraph{Multi-objective problems} Window games have also been considered in the multidimension setting, where several weight (resp.~priority) functions are given, and the objective is defined as the intersection of all one-dimension objectives~\cite{DBLP:journals/iandc/Chatterjee0RR15,DBLP:journals/corr/BruyereHR16}. Again, \textit{our generic approach supports effortless extension to this setting}.

In the direct case, the unfoldings of Sect.~\ref{sec:fixed} can easily be generalized to multiple dimensions, as in~\cite{DBLP:journals/iandc/Chatterjee0RR15,DBLP:journals/corr/BruyereHR16}. For prefix-independent variants, the EC classification needs to be adapted to handle multidimension window games, which we can solve using the techniques of~\cite{DBLP:journals/iandc/Chatterjee0RR15,DBLP:journals/corr/BruyereHR16}. Then, we also need to consider a multi-objective reachability problem~\cite{DBLP:journals/fmsd/RandourRS17}.
While almost all cases of multidimension window games are $\mathsf{EXPTIME}$-complete, note that the decidability of the bounded mean-payoff case is still open; it is however known to be non-primitive recursive hard.

\paragraph{Tool support} Thanks to its low complexity and its adequacy w.r.t.~applications, our window framework lends itself well to tool development. We are currently building a tool suite for MDPs with window objectives based on the main results of this paper along with the aforementioned extensions. Our aim is to provide a dedicated extension of \textsc{Storm}, a cutting-edge probabilistic model checker~\cite{DBLP:conf/cav/DehnertJK017}.

\bibliographystyle{alpha}
\bibliography{main}

\end{document}